\documentclass[11pt]{article}
\usepackage[margin=1in]{geometry} %
\usepackage{siunitx}
\usepackage{amsmath, amsthm, mathtools, amssymb}
\usepackage{booktabs}
\usepackage{tabularx}
\usepackage{thm-restate}
\usepackage{verbatim}
\usepackage{changepage}
\usepackage{algorithm}
\usepackage{algorithmic}
\usepackage{url}
\usepackage[caption=false]{subfig}
\usepackage[square,numbers,sort&compress]{natbib}
\renewcommand{\cite}[1]{\citep{#1}}
\usepackage[final]{graphicx}
\usepackage[hidelinks]{hyperref}

\usepackage[english]{babel}
\usepackage{multirow}

\usepackage{cleveref}
\usepackage{color}
\usepackage{soul}
\usepackage{float}

\theoremstyle{plain}
\theoremstyle{definition}
\newtheorem{theorem}{Theorem}[section]

\newtheorem{lemma}[theorem]{Lemma}
\newtheorem{assumption}[theorem]{Assumption}

\newtheorem{proposition}[theorem]{Proposition}

\newtheorem{definition}[theorem]{Definition}

\newtheorem{exmp}{Example}[section]
\theoremstyle{remark}

\title{Optimal Strategies in Ranked Choice Voting}
\author{Sanyukta Deshpande\thanks{University of Illinois at Urbana-Champaign} \and Nikhil Garg\thanks{Cornell Tech} \and Sheldon H. Jacobson\footnotemark[1]}

\date{}

\begin{document}
\maketitle

\begin{abstract}
Ranked Choice Voting (RCV) and Single Transferable Voting (STV) are widely valued; but are complex to understand due to intricate per-round vote transfers. Questions like determining how far a candidate is from winning or identifying effective election strategies are computationally challenging as minor changes in voter rankings can lead to significant ripple effects-- for example, lending support to a losing candidate can prevent their votes from transferring to a more competitive opponent. We study optimal strategies — persuading voters to change their ballots or adding new voters — both algorithmically and theoretically. Algorithmically, we develop efficient methods to reduce election instances while maintaining optimization accuracy, effectively circumventing the computational complexity barrier. Theoretically, we analyze the effectiveness of strategies under both perfect and imperfect polling information.  Our algorithmic approach applies to the ranked-choice polling data on the US 2024 Republican Primary, finding, for example, that several candidates would have been optimally served by boosting another candidate instead of themselves.
\end{abstract}

\section{Introduction} \label{sec: intro}

Ranked Choice Voting (RCV), in which the voters submit a ranking over candidates,  is promised to increase representation proportionality, curtail gerrymandering, reduce the need for strategic voting, and increase campaign civility \cite{FairVoteRCV}. For such reasons, RCV is used in several countries including Australia, Ireland, and New Zealand. In the United States, about 50 jurisdictions use it as of Nov 2023, reaching approximately 13 million voters \cite{FairVoteRCVUSE}, with more planned. Several jurisdictions use RCV to elect multiple winners (such as for city councils), using Single Transferable Vote, which has further been supported at a national level, e.g., to elect representatives in multi-member districts in a proposed bill, `Fair Representation Act, H.R. 4000, 2024' \cite{fra}.

RCV and STV operate as follows, over multiple ``rounds'' calculated automatically given voters' ranked lists over the candidates. In each round, candidates surpassing a certain quota (e.g., the Droop quota) secure victory, with surplus votes (votes beyond the quota) transferred to subsequent choices (through some transfer rule). If no candidate meets the quota, the candidate with the fewest votes is eliminated, and all of their votes are transferred. This process continues until either all available seats are filled or the number of remaining candidates matches the available seats. (For example, in single-winner RCV, i.e., Instant Runoff Voting (IRV), candidates are eliminated in each round until a candidate surpasses $50\%$.) %

The transfer of votes as candidates are eliminated is cited informally as the rationale for both increased campaign civility (candidates want to be the second preference of voters who prefer other candidates) and the reduced need for strategic voting (automatic transfer of votes is communicated as meaning that voters do not have to worry as much about supporting a ``spoiler'' who prevents a win of their second-most preferred candidate). For example, candidates in RCV and STV elections often endorse one another, asking their voters to vote in ``blocs.''
\begin{quote}
    \setlength{\leftskip}{-2em}\setlength{\rightskip}{-1em}- \textit{In Tight Harlem Race, Pair of Candidates Say ‘Rank Us’ to Push Out Third. } (The City, \citet{Honan2023Harlem})
\end{quote}
\begin{quote}
    \setlength{\leftskip}{-2em}\setlength{\rightskip}{-1em}- \textit{Five candidates running for a hotly contested City Council seat in Upper Manhattan cross-endorsed each other on Thursday, a rare move that is intended to take advantage of the city's new ranked-choice voting system.} (Patch, \citet{Garber2024Harlem})
\end{quote}

\noindent Are these rationales ``true,'' both theoretically and empirically? Theoretically, of course, RCV is not strategy-proof---e.g., candidates can serve as ``spoilers'' (their presence can lead to the elimination of a candidate who would win without them) and even monotonic improvement in a candidate's position in voters' rankings, as would occur with cross-endorsements, may worsen a candidate's overall ranking; \citet{graham2023examination} find that such concerns are empirically rare, but not non-existent (Also see \citet{graham2022mathematical, kilgour2020prevalence}).

Given this ambiguity, how should a candidate proceed? Should they pursue cross-endorsement coalitions, or seek to support another candidate who can serve as a spoiler to their chief competitor? More importantly, how far is a candidate from winning? How robust is the current most likely outcome to small changes in voter turnout or persuasion campaigns? Which candidates are ``irrelevant'' in the election, in the sense that no reasonable amount of additional votes (of arbitrary ranked lists) would help them win or even cause them to serve as a spoiler? 

Such questions are relatively easy to answer given polling data in single-round plurality voting, where a candidate simply needs more votes than other candidates. With RCV and STV, however, the fact that votes transfer across multiple rounds induces substantial empirical and computational complexity -- there is an exponentially large number of ways (in the sense of the round-by-round outcomes) a final ranking may be reached. This difficulty may limit the clarity of RCV and hinders adoption \cite{RCVNightmare2024, donovan2019self}. \citet{clark2020rank} finds in Maine that the negative effects on voter confidence, satisfaction, and ease of use outweighed the potential positives, with no improvement in campaign civility.

As an example, \Cref{tab: sf_election} shows the distribution of ballots in the San Francisco District 7 2020 RCV election after all but the final three candidates were eliminated. Engardio (who ultimately placed second despite having the most first-rank votes) would need 2192 additional first-rank votes to win; however, Engardio would also win if, instead, Nguyen got 798 more votes. In hindsight, it would have been more ``efficient'' for Engardio to help increase Nguyen's support, instead of their own.\footnote{We provide two more examples of how even small-scale strategies can affect elections in Appendix \ref{app: RCV_examples}. In Example \ref{ex: campaign}, a candidate wins by adding votes to another candidate who eventually loses. In Example \ref{ex: campaign_flip}, the entire social choice order is flipped as one voter updates their vote.} Ex-ante, however, it may have been risky for Engardio to try to increase Nguyen's support: help the ``spoiler'' too much, and they may instead win the election.

\begin{table}[tb]
\centering
\begin{tabular}{l|c|c|c|c|c|c|c|c|c}
\textbf{Num. Voters:} & 5237 & 3316 & 5566 & 4050 & 5708 & 1894 & 2251 & 6909 & 1695 \\ \hline
\textbf{Rank 1}  & [E]ngardio & E & E & M & M & M & N & N & N \\ \hline
\textbf{Rank 2}  & [M]elgar & N & - & E & N & - & E & M & - \\ \hline
\textbf{Rank 3}  & [N]guyen & M & - & N & E & - & M & E & - 
\end{tabular}
\caption{The 2020 District 7 Board of Supervisors election in San Francisco, CA, \cite{graham2023examination}.  Counting the first choice votes, we have Engardio: 14119 votes (the sum of 5237, 3316 and 5566), Melgar: 11652, and Nguyen: 10855. Thus, Nguyen is eliminated next. Then, Melgar wins with 18561 to Engardio's 16370. If Melgar is eliminated first instead of Nguyen, post transfers Engardio is at 18169 and Nguyen is at 16563. This implies that Engardio needs 2192 votes to win, but the same is achieved if Nguyen gets 798 more votes. %
Hence, supporting a rival candidate is a $63\%$ more efficient manner for Engardio to win.
}  \label{tab: sf_election}
\end{table}

We study these questions, theoretically, algorithmically, and empirically -- given knowledge of the state of the race, how can a candidate or coalition optimally reach their desired result?
At the base of our analysis is a characterization of RCV and STV in terms of the social choice order (ranking over candidates, what we call ``orders'') they produce, \textit{and} the round-by-round results (array of result outcomes, i.e., eliminations or wins for each round, what we call ``sequences'') that produced that result. Using this framework, we study optimal strategies -- \textit{persuading} voters to change their ballots or \textit{adding} new voters -- both theoretically and algorithmically. (With optimality in the sense of the least number of changes required.) We do so assuming either perfect (when the status quo of the race is perfectly known) or imperfect polling information. Our contributions are as follows.

\textbf{Algorithmic framework. } 
We first establish that an optimal vote addition strategy to reach a \textit{given} structure (final order \textit{and} round-by-round outcome sequence) given status quo vote data can be calculated in polynomial time. Unfortunately, with $n$ candidates, a given candidate is a top-$k$ winner in $k \times (n-1)!$ different orders, and each order naively has $2^{n-1}$ compatible outcome sequences -- a naive search would require analyzing $k \times (n-1)! \times 2^{n-1}$ structures, to find an optimal winning strategy. 

We then show how to \textit{circumvent} this computational hardness in practice by \textit{reducing the election instance}, given the vote data, \textit{without affecting the optimality of the calculated strategies} and in polynomial time. (a) We algorithmically reduce the candidate space by removing candidates who will remain irrelevant even if $B$ new votes of arbitrary ranked lists are added. (b) We further reduce the compatible per-round outcome sequences for a given order, both in a data-independent and dependent manner. While these reductions do not reduce the worst-case computational complexity of calculating optimal strategies for arbitrary vote data (and in fact, cannot do so, given known hardness results \cite{bartholdi1991single}), they may practically do so in many settings.

\textbf{Empirical case study.}\footnote{Our code and data are available at: \url{https://github.com/sylviasolo/Optimal_Strategies_in_RCV}} We illustrate our framework on a real-world dataset: ranked choice poll data from the 2024 United States Republican Presidential Primary, conducted in September 2023 by FairVote and WPA Intelligence \cite{OtisLaverty2023}. Using our algorithmic framework, we show that for reasonable strategic additions, we can remove at least 7 candidates (out of 13), without affecting resulting optimality. Alongside the reduction in the number of per-round results to search through, the runtime to find optimal strategies is reduced to 21 seconds (on a modern laptop) compared to an estimated 1186 days without such reductions. This further enables \textit{bootstrapping} the polling data, to evaluate strategy robustness under imperfect information. 

Our analysis provides insight on the Republican Primary and on strategies in RCV voting generally. For example, we find that while Haley was robustly in second place in the RCV outcome, she would need more new first-place votes to overtake Trump than would DeSantis or Ramaswamy. Furthermore, if the polling data \textit{exactly} represented the state of the race, it would be more efficient for Ramaswamy or DeSantis to overtake Haley by adding votes to Christie to serve as a spoiler, instead of just adding first-place votes for themselves. However, this strategy would be risky in the presence of uncertainty, as it may instead lead to another candidate winning. 

\textbf{Theoretical analysis.}  
Given the insights from the case study, we further characterize optimal strategies (adding or changing votes), under perfect or imperfect information. Under perfect information, we first present a result: coalitions (including a singleton) cannot mutually benefit if candidates cannot influence a voter's preferences over candidates that the voter prefers over the strategizing candidates. We then expand the candidate's action space, so that they can add (a limited number of) arbitrary ranked list votes to the election. We show that even in this setting, in standard cases the most efficient strategies are to enable spoilers -- to add ballots that highly rank other candidates that will eventually lose and transfer votes to the strategizing candidates. Finally, we study such strategies under imperfect information. In this case, forming coalitions can be ex ante beneficial to all candidates; thus, a rationale for coalition building comes from imperfect information. However, for a given candidate, the only vote addition strategy guaranteed to be beneficial ex-post is to add first-place votes to oneself.

~\\\noindent Putting things together, our results provide evidence for both sides of the debate on undesirable strategic behavior by candidates in ranked-choice voting. While under perfect information it is often optimal (both theoretically and empirically) to support spoilers, such strategies are risky and may backfire under uncertainty in practice. Our algorithmic framework further provides a tractable approach to analyzing RCV, informing media coverage and campaign strategies.

\subsection{Related literature}
It is well known that RCV and STV are not strategy-proof, and that there exists no social welfare function that is strategy-proof and satisfies other reasonable axioms \cite{arrow2012social,gibbard1973manipulation, satterthwaite1975strategy}. Thus, strategic behavior in ranked-choice voting has been extensively studied. %

Theoretically, it has been explored from various angles, including computational complexity, combinatorial voting, and axiomatic approaches (\citet{ brandt2016handbook} provide an extensive overview).  For various voting protocols, strategic behavior is modeled from voters' perspectives, as well as under control, i.e., adding or removing votes by an election organizer, and bribery, i.e., changing existing votes \cite{bartholdi1992hard, meir2008complexity, faliszewski2009hard,  xia2010scheduling,faliszewski2010using}.  For STV, it is known that finding strategic voting manipulations for an unbounded set of candidates, constructive (making someone a winner) or destructive (making someone a loser) coalition-weighted manipulations for more than 2 candidates, or computing margin of victory (MOV) are all NP-complete \cite{bartholdi1991single, conitzer2007elections, xia2012computing}. 
The concept of manipulation has been broadened by \citet{dutta2001strategic} to encompass strategic candidacy, considering candidates' motivations for participating in elections \cite{brill2015strategic}. %
While this literature largely focuses on hardness results, we develop approaches to reduce the election instance while maintaining optimality, to attempt to circumvent such results in practice. 
Although our setup may be extended to various manipulation rules, we concentrate on strategic voter additions as our primary algorithmic focus.

Empirical evidence suggests the feasibility of manipulating STV elections \cite{walsh2010empirical}. The existence of undesired traits such as failure to elect Condorcet winner, upward and downward monotonicity, truncation and spoiler effects, etc, on real data has been empirically studied \cite{graham2023examination}. Studies have investigated ballot length, highlighting the impact of ballot exhaustion \cite{dickerson2023empirical, baumeister2012campaigns, ayadi2019single, tomlinson2023ballot}. The spoiler effect, which describes the influence of less-favored candidates on the outcome, is shown to be significant in RCV \cite{mccune2023ranked}. 

Most related to our work are those that algorithmically study possible manipulations, under various models. The concept of the `margin of victory' (MOV), the smallest manipulation possible, was discussed in relation to IRV \cite{magrino2011computing}. Expanding upon this, \citet{blom2019toward, blom2020did} developed algorithms to calculate the lower and upper bounds of MOV in multi-winner STV, noting that the developments are limited in scope owing to complexity. %
In the case of IRV, \citet{jelvani2022identifying} developed algorithms to find the possible set of winners with outstanding ballots. Their approach involves brute force on IRV, without claiming optimality, and discarding candidates that have less than 5 percent votes, as it is computationally difficult to analyze large candidate sets. Theoretically, such removal may have unseen ripple effects. %

Our work extends this body of literature by presenting a framework that finds \textit{optimally} efficient strategies for candidates. While the original optimization problem is prohibitive with many candidates, we demonstrate how to exploit its structure to efficiently characterize the influence of voter data updates, leading to a practical advance without sacrificing optimality. Overall, our work has improved over the previous contributions in three dimensions:  generalization of the system (IRV to STV), reduced complexity of solving, and the optimality of solutions.

\section{Model} \label{sec: model}
We begin by defining ranked-choice voting (RCV) system with single transferable vote (STV) and setting up combinatorial foundations of our model. We first formulate STV in a manner that leads to a well-defined social welfare function $\mathcal{F}_{STV}$ in which, in each round, exactly one candidate is declared a round-loser or a round (and election) winner. This formulation enables a characterization of outcomes in terms of ``orders'' (the final ranking, or ``social choice order,'' over candidates) and ``sequences'' (each round's outcome in terms of the round-loser or winner). Strategies in our theoretical and algorithmic framework operate on these structures.

\subsection{Notation and Framework}\label{sec: setup}

An \textit{election instance} consists of a set of $n$ candidates $\mathcal{C}$ and a set of voter ballots $\mathcal{B}$ that rank the candidates. Each voter ballot in $\mathcal{B}$ is an order, partially or fully ranking candidates in $\mathcal{C}$.  Suppose that we are electing $k$ winners, using Single Transferable Voting, defined in detail below. In each round, a candidate is either declared a winner (if they have a number of votes at least as large as the Droop quota $Q$, defined as $1+\frac{|\mathcal{B}|}{(k+1)}$), or the candidate with the fewest votes is eliminated.

Let $I = \{a, b,  \dots, j\} \subset  \mathcal{C}$ denote a set of candidates, with $i = |I|$. We use variable $V^i_{a, b, \dots, j}$ to denote the total number of voters who rank the candidates in $I$ at the top of their ballot, in order $a>b>\dots >j$. The number $V^i_{a, b,\dots, j}$ is aggregated, i.e., it includes voters who may rank any other candidates after $j$ (in any order), although they strictly follow the order for the candidates in $I$. The superscript $i$ is useful in bookkeeping the number of candidates being ranked. %
The variables $\{V^i_{a, b,\dots, j}\}_{I \subset \mathcal{C}}$, together capture all information about the ballots in $\mathcal{B}$. %

\subsubsection*{Defining STV} 

We now describe the STV mechanism formally in our notation, as a social choice function that maps variables $\{V^i_{a, b,\dots, j}\}_{I \subset \mathcal{C}}$ into a unique outcome ordering over the candidates. Let $F$ be the set of orders, i.e., corresponding to the $n!$ possible outcome orders over the $n$ candidates. An outcome order $f \in F$ is constructed from the ballots as follows. At a high level, the mechanism places a round-winning candidate at the top-most available position, if a win is possible. If not, it places the round-losing candidate at the bottom-most available position, thereby removing one candidate from the contest in each round. This process terminates in $n$ rounds, producing the order $f$. In the election, the first $k$ candidates in the order $f$ are deemed as election winners.\footnote{Our definition is equivalent to other standard definitions of STV protocol, such as those that allow multiple candidates to be elected in a single round or eliminate candidates until $k$ candidates surpass the Droop quota. Such definitions may not produce a strict
 a social choice order, as winners chosen simultaneously are tied in the order. Our algorithmic approach requires distinguishing between such tied candidates, as we consider vote changes that move from one order to another. In \Cref{app: protocolproof}, we show that the mechanism is well-defined, as it maps each election scenario uniquely to its corresponding order, yielding the same set of winners as such a standard definition. 
}

\begin{enumerate}
    \item Let $f$ denote the current order, initalized as an empty array; and $C \subset \mathcal{C}$ as the set of active remaining candidates, initialized as $C = \mathcal{C}$.

    \item Define $\{\Tilde{V}^i_{a,b, \dots, j}\}_{I \subset  C}$ by processing transfers from out-of-contest candidates to in-contest candidates, if any.
    Order all candidates in $c\in C$ according to the number of votes $\Tilde{V}_c^1$. Check if the top candidate $T$ gets more than the Droop quota, i.e. if $\Tilde{V}^1_T\geq Q$. In the case of ties, choose the candidate according to a pre-determined tie-breaking order.
    \item  If yes, $T$ is the next winner, is added to the top available position in order $f$, and is removed from the set of active candidates $C$. Transfer their \textit{surplus votes} (number of votes minus the number of votes needed to win Droop quota, i.e., $\Tilde{V}_T^1$ - $Q$) to their subsequent choice candidates. 
    
    We use the fractional transfer rules: each remaining candidate $R \in C$ gets surplus votes weighted by its share of votes $\Tilde{V}_{T,R}^2$ in the total votes contributing to the win, i.e., a surplus of $\frac{\Tilde{V}_{T,R}^2}{\Tilde{V}_{T}^1}(\Tilde{V}_T ^1- Q)$ is transferred to $R$. %
    
    This is the Weighted Inclusive Gregory method \cite{PRF2024WIGM} for STV, although our results apply to any deterministic STV surplus transfer method.
    \item  Otherwise, If there does not exist such a winning candidate, place the contest bottom candidate $B$ (the next loser) to the last available position in $f$. Transfer all of its votes, including those received from transfers, to their subsequent choice candidates, if any exist.  
    \item %
    Restart with Step 2 until $C$ becomes empty and $f$ becomes full. The resulting order in $f$ is the aggregated social choice order.
\end{enumerate}
Denote the above function as 
$\mathcal{F}_{STV}: \mathcal{B} \mapsto F$, i.e., $%
   \mathcal{F}_{STV} (\{V^i_{a, b,\dots, j}\}_{I \subset \mathcal{C}}) = f.$

\subsection{Structural properties of the social choice function}
We now develop structural properties of the social choice function $\mathcal{F}_{STV}$, that will serve as the foundation of our analysis and algorithmic framework.

First, note that STV outcomes (both per-round outcomes of whether a candidate is elected or eliminated, and the resulting overall order) are robust to small perturbations (unless there are ties or a candidate receives votes equal to the Droop quota $Q$ in any round); for each ballot set, there is an open set in the ballot space, all leading to the same per-round outcome and overall order -- i.e., preserving the pattern of wins ($W$) and eliminations ($L$) for each round.

This observation allows us to partition the ballot space into regions, referred to as \textit{orders}, each defining a unique social choice order $f \in F$ obtained via $\mathcal{F}_{STV}$. There are $n!$ orders in total, for $n$ candidates. Further, each order may correspond to enumerable \textit{sequences}, each defined as an array of round-specific outcomes, i.e., describing the win (W) or elimination (L) outcome of each round. Together, this defines a \textit{structure} that has two components- an order and a sequence- which jointly capture the STV election proceeding. For example, an `order' may be the collective ranking $(A>B>C>D)$ of 4 candidates, with an associated `sequence' like $(W, L, W, W)$. This formulates the structure $[(A>B>C>D), (W, L, W, W)]$, in bijection with the election proceeding [$A$ wins, $D$ loses, $B$ wins, $C$ wins] in 4 rounds. There are $2^{n-1}$ sequences that an order may correspond to, as each round except the final one can yield two outcomes.  See Example \ref{ex: structures}  for an illustration of structures in an election with 4 candidates. 

\begin{exmp} \label{ex: structures}
Consider $\mathcal{C} = \{A, B,C,D\}$ as the set of contesting candidates. We show the social choice order, $A>B>C>D$, with $8$ possible sequences leading to 8 different structures as given in Table \ref{tab: sub_st_example}.
\begin{table}[ht]
\centering
\small
\resizebox{\textwidth}{!}{%
\begin{tabular}{|c|c|c|c|c|c|c|c|c|}
\hline
Round & (W,W,W,W) & (W,W,L,W) & (W,L,L,W) & (W,L,W,W) & (L,W,W,W) & (L,W,L,W) & (L,L,W,W) & (L,L,L,W) \\ \hline
1     & A wins    & A wins    & A wins    & A wins    & D loses   & D loses   & D loses   & D loses   \\ \hline
2     & B wins    & B wins    & D loses   & D loses   & A wins    & A wins    & C loses   & C loses   \\ \hline
3     & C wins    & D loses   & C loses   & B wins    & C wins    & C loses   & A wins    & B loses   \\ \hline
4     & D wins    & C wins    & B wins    & C wins    & B wins    & B wins    & B wins    & A wins    \\ \hline
\end{tabular}%
}
\caption{Structures associated with the order $A>B>C>D$.}
\label{tab: sub_st_example}
\end{table}
\normalsize
\end{exmp}

\Cref{thm: combi_data} formalizes the above discussion, defining enumerable non-linear constraints that partition the ballot space into orders and sequences. %

\begin{proposition}
 \label{thm: combi_data}
With suitable tie-breaking, each voter data $\mathcal{B}$ corresponds to a single structure (a pair of an order and a sequence).
For each structure, a set of non-linear constraints involving variables $V^i_{a, b,\dots, j}$ defines a distinct region in the voter ballot space where the election outcome remains constant. There are $n!$ possible orders, i.e., the collective social choice rankings, and $2^{n-1}$ sequences, i.e. per round elimination or winner information.

\end{proposition}

We note that the election result can be obtained either by running the STV mechanism or by verifying constraint satisfiability: with the rules of breaking ties in place, each point in the space satisfies the constraints of exactly one structure. Generating the constraints that define the viability of any structure may be achieved algorithmically. %
(We show how to construct such constraints in  Appendix \ref{app: algo_constraints}, using Algorithm 
\ref{algo: round_update} and \ref{algo: constraints}). Our theoretical proofs and algorithmic framework use such constraints. 

Next, we center the discussion on election strategies, allowing candidates to compete to be among the top $k$ in a social choice order. Specifically, candidates or coalitions can strategize to minimize the voter manipulation required to achieve a more favorable position within the main-structure $f$. Due to the cascading nature of RCV elections, candidates may employ both cooperative and competitive strategies to their benefit. 

Our focus lies on a specific form of manipulation: strategic campaigning, which involves the addition of new (arbitrarily ranked) votes, e.g., using the Get-Out-The-Vote (GOTV) campaigns. 
Campaigning is generally recognized as a legitimate and subtle way to update voter data, contrasting with strategic voting, i.e., direct manipulation of voter preferences\footnote{GOTV effects are appropriate to study as it is noted that campaigning, a core part of election strategies, rarely changes one's vote, but affects the overall voter turnout \cite{enos2018aggregate, deshpande2023votemandering}.}. Yet, the implications of strategic campaigning may be translated to strategic voting as well: if a strategy involves the addition of $X$ votes, the effect may be achieved by altering fewer preferences of current voters. 

The fact that even a single order has an exponentially large number of sequences is what makes optimal strategies in RCV or STV especially hard to find: strategies that seek to produce a different order by changing the set of ballots can do so by inducing any of an exponentially large number of sequences that correspond to the target order. 
Each sequence induces a different transfer of votes per round and affects future round outcomes;   this information is important: the strategies starting from a sequence $(W,L,W,W)$ may include affecting the first choice votes to win, while the same for $(L,W,L,W)$ may involve helping the first losing candidate avoid elimination, to manipulate the transfer votes.  Given this complexity, we first develop an algorithmic framework to calculate optimal strategies (\Cref{sec: cut_search_space}) that is efficient in practice (\Cref{sec: casestudies}), and then theoretically derive properties of such optimal strategies, broadening manipulations to also contain cross-endorsements (\Cref{sec:theory}).

\section{Algorithmic Framework} \label{sec: cut_search_space}
We now develop a computational framework to find optimal vote addition strategies for each candidate under perfect information -- if we want a candidate to place in the top $k$ of the election, what is the optimal (requiring the minimum number of additional ranked ballots) vote addition strategy? Our algorithmic framework is as follows, given $n$ candidates and $m$ unique voter ballots.

(1) Given a budget of $B$ additional votes, we first in \Cref{thm: poly_efforts_B} develop an algorithm to optimally (if possible) reach a given structure in $O(mn)$ time. Then, binary search on the budget $B$ yields the optimal strategy to efficiently reach a given structure. 

However, with many candidates, there are a prohibitively large number of structures: there are $n!$ possible orders, of which the candidate would be in top $k$ in $k \times (n - 1)!$ orders. Each order has $2^{n-1}$ sequences. Naively, then, finding an optimal strategy requires finding the minimum vote additions over $k \times (n - 1)! \times 2^{n-1}$ structures.  

(2) Thus, we develop approaches to reduce the election size, \textit{without} affecting the optimality of the calculated strategies. (a) In \Cref{thm: remove_irrelevant_candidates}, given a budget of $B$ additional votes and status quo vote data, we give an algorithm with 
$O(mn^4)$ complexity that removes a set of irrelevant candidates who will be eliminated first regardless of how the $B$ votes are added. 
(b) In \Cref{cor: sequences}, we first show that for any set of $k$ winners there are only 
$\sum_{j=1}^k \binom{n}{j}$ feasible sequences. 
Then, we show how to reduce the number of sequences further given status quo vote data: in \Cref{thm: substtheorem}, we give an algorithm with $O(mn^2)$ complexity that reduces the number of feasible sequences that can lead to an optimal win. Reducing the search space is important: as detailed in \Cref{sec: casestudies}, our case study becomes computationally feasible only due to these reductions.

\subsection{Optimizing to reach an outcome structure}\label{sec: polytime_algo}

We first develop an approach to optimally add votes that would change the current structure (per-round elimination or winner information) into a given desired structure.\footnote{Note that any structure is feasible given enough votes.  Intuitively, a disproportionately high number of votes may always be added to any candidate, per the desired structure. This makes the original voter ballot data irrelevant in comparison, allowing us to reach any structure.} Recall that adding a vote {could} mean adding a full ballot that ranks all $n$ candidates, or adding a shortened ballot. Our approach minimizes the total number of votes added, while keeping ballot lengths minimal. %

Note that adding $B$ votes has two effects: it adds votes to given candidates in each round, and it changes the Droop quota $Q$ required to win in each given round. Thus, we require a two-step approach: (a) for each given budget $B$ (fixing the Droop quota), we calculate an optimal set of additions of \textit{up to} $B$ votes to reach the desired structure. (b) Then, we embed this algorithm within a binary search over the budget $B$, until the actual number of votes used equals the budget. 

Theorem \ref{thm: poly_efforts_B} shows that the first stage can be done in polynomial time, and proposes a smart allocation algorithm [\textsc{SmartAllocation} in \Cref{app:allocation_algo}] that outputs the optimal strategic allocation if any exists for using up to $B$ additional votes. %

\begin{restatable}{theorem}{polyefforts} \label{thm: poly_efforts_B}
    Given budget $B$, an optimal strategic addition (if it exists) to reach any specific structure (order and sequence pair) is computable in polynomial time \(O(mn)\), with $n$ candidates and $m$  unique ballots.
\end{restatable}
\begin{proof}[Proof Sketch]
We discuss the proof idea here. For the full allocation algorithm [\textsc{SmartAllocation}] and details on the proof, see \cref{app:allocation_algo}.
Starting with the first round, we iteratively compute the slack of votes for candidates in each round (i.e., the number of extra votes each candidate needs to establish the required structure for that round -- either to win or to not get eliminated) and allocate votes accordingly. From the second round onward, we also keep an account of previously allocated votes which become `available' as prior beneficiary candidates win or get eliminated, leading to transfers. Before any new allocation of votes, the `available' votes are used to fulfill slacks, by adding candidates to those existing ballots. %
The algorithm terminates at the last round or when the budget is exhausted. %

The algorithm adds votes to consecutively fill the round slacks, although the election is processed on the data comprising additions made over all rounds. Thus, filling slacks of one round may affect other rounds too, potentially creating new slacks or reaching infeasibility.  In the proof, we thus argue that the algorithm ensures the following to maintain the optimality and convergence:  (a) the process of finding smart additions per round guarantees optimality of allocation over the entire election instance (which we show using an induction argument); (b) no new slacks are generated after all slacks are filled (derived from the structure framework and the allocation rule). %
\end{proof}

While the possible budget $B$ fixes the Droop Quota, it may be that the optimal ballot additions to establish each round result does not require the full $B$ votes (suppose that it instead uses $b$ votes). We thus embed the above algorithm within a  binary search to optimize the outputted strategic allocation, until convergence. Decreasing the budget $B$ may lead to a lower $b$ value if vote additions are made to win a round, as a lower Droop quota implies a lower target for winning. However, decreasing the budget may reach infeasibility if the addition of votes increased the Droop quota -- and was useful in changing the status of a winning round to an elimination round. We thus stop decreasing the value of $B$ when $b=B$ or if we hit infeasibility. In the latter case, the resulting optimal strategy needs to fill in $B - b$ empty ballots. %

\subsection{Efficiently navigating the search space} \label{sec: efficient_trim}

The above approach (including the binary search with a maximum budget of $B$), requires runtime $O(mn \log B$), \textit{for each structure} (round-by-round outcome characterization). However, for each given candidate, since there are exponentially many structures in which that candidate is among the top $k$ winners out of $n$ candidates, the overall algorithm to find optimal candidate vote addition strategies would be prohibitively expensive with many candidates. We thus present two approaches to reduce this runtime, by reducing the space of candidates and the number of structures, respectively. Note that our approaches are not \textit{heuristics}: they do not affect the optimality of the calculated strategies. 

\subsubsection{Reduction of the candidates set}
Our first approach starts with the observation that, in practice, STV election results often start with multiple eliminations of candidates who have received comparatively much fewer first-choice votes than others.\footnote{For example, we note this in data from New York City RCV elections \cite{NYCBoardOfElections2024}} 
Given a budget of $B$ (arbitrary) additional votes, we show how to reduce the set of candidates without affecting later-round dynamics. After the reduction, the corresponding votes of all such candidates are transferred to their subsequent choices (if any choices remain), without affecting resulting proceedings. This allows us to form an equivalent representation of the election instance, with fewer candidates and an updated ballot set (with fewer unique ballots). With this, we may optimize for strategies via fewer structures, significantly shortening the search space. 

\begin{restatable}{theorem}{candidateremoval} \label{thm: remove_irrelevant_candidates}
    Given a total addition of $B$ undescribed ranked-choice votes,  the election instance may be shortened by removing irrelevant candidates who will always be eliminated first regardless of which votes are added. This is achieved efficiently, with $O(mn^4)$ complexity for $m$ number of unique ballots and $n$ candidates in the original instance.
\end{restatable}
\begin{proof}[Proof Sketch]
We sketch the idea here and refer to Appendix \ref{app: cut_search_space} for details. First, suppose that, up to the $j$th round, we have only had candidate \textit{eliminations} and not any winners; then, the \textit{order} of eliminations is irrelevant in the transfer votes after that round (Lemma \ref{lem: set_reduction}). This implies that if we remove the set of candidates occupying strictly consecutive lower positions in the social choice order, and run the election for the smaller set, the \emph{remaining candidates}
will appear in the same order as before.
We build such a group of candidates under arbitrary additions of $B$ votes. 

\textbf{The main idea:} A group from the bottom, say $R \subset \mathcal{C}$, can be removed if, for each candidate $C_i \in R$, the following holds: [the votes that $C_i$ has after the removal and subsequent transfer of all others in $R\setminus \{C_i\}$ + $B$ votes]  is less than [the first choice votes of $G$ + transfers from $R \setminus \{C_i\}$] for all candidates in $G \in \mathcal{C} \setminus R$. This ensures, without going through all possible combinatorial orders of $R$, that given any elimination order of the candidates in $R$, they all are still short of more than $B$ votes to compete with any remaining candidates.  The algorithm then iteratively builds a group of candidates $R$ until the condition gets satisfied, and removes it, transferring their votes to subsequent candidates still in the election.  After removal, the process restarts, and multiple such groups may be eliminated this way. 
\end{proof} %

We note that the number of candidates who can be eliminated in such a way is highly dependent on the budget $B$ and the current set of votes (for example, if the original instance is such that all candidates are within $B$ ballots of each other in \textit{all rounds} that they are active in, then no candidates may be eliminated). %
In our case study in \Cref{sec: casestudies}, this algorithm eliminates at least $7$ candidates out of $13$ in $98\%$ of the bootstrap data samples of the polling data for $B=5\%$, allowing effective optimal strategy computation in an otherwise close to intractable scenario. %

\subsubsection{Reducing the space of sequences}

For a multi-winner election, the theoretical limit on the number of sequences is exponentially high ($2^{n-1}$), as given in Proposition $\ref{thm: combi_data}$. However, we now show that, independent of the status quo votes and budget, not all the sequences are feasible when we seek to elect $k$ winners, due to how the Droop quota is defined:

\begin{proposition} \label{cor: sequences}
    The number of feasible sequences is at most $\sum_{j=1}^k \binom{n}{j}$. %
\end{proposition}
\begin{proof}[Proof Sketch]
    Appendix \Cref{lem: round-election} establishes that a round-eliminated candidate may be an election winner, but a round-winning candidate is never an election loser, as winning a round implies a number of votes above the Droop quota. Thus, the number of round-winning candidates is at most $k$, as only $k$ candidates win the election. Then, at most $k$ winning positions may be occupied at $n$ positions in the social choice order $f$, and the remaining positions are eliminations. Finally, the number of wins in a sequences ranges from 1 to $k$, and so the number of feasible sequences is $\sum_{j=1}^k \binom{n}{j}$. %
\end{proof}

We can also further reduce the number of sequences without affecting optimality, using the following intuition. In practice, some sequences (i.e., whether a candidate was eliminated or a winner declared in each round) are rarely attained, as it is often less costly to move across orders (i.e., the final ordering over candidates) than it is to move across sequences while keeping the same order. This is because changing the sequence involves at least one change of round result type (either $W \to L$ or $L \to W$) and changing the order involves replacing a win or a loss of one candidate with another. Additionally, if the election instance contains significant partially ranked ballots, the pool of active votes decreases as with the eliminations, while the Droop quota, the bar for a round-win, remains the same. %
Using these ideas, we develop an algorithm to remove suboptimal sequences under the addition of $B$ arbitrary votes. %

\begin{restatable}{proposition}{substtheorem}\label{thm: substtheorem}
    Given a total addition of $B$ undescribed ranked-choice votes to the voter base, the search space of sequences may be shortened, with complexity $O(mn^2)$ for $m$ number of unique ballots and $n$ candidates.
\end{restatable}
\begin{proof}[Proof Sketch]
Here we use the sequence of the given voter data to find attainable sequences, given an addition of $B$ undescribed votes to it. We first find an upper bound on the total number of winners possible by computing the number of unique ballots that can contribute to wins, and dividing it by the Droop quota. We then find a lower bound on the number of consecutive initial losses, measuring the minimum accumulation of votes through transfers that can lead to the first win. See Appendix \ref{app: cut_search_space} for details on both these bounds. 
\end{proof}

\section{2024 US Republican Primaries Case Study} \label{sec: casestudies}

\noindent We now use our framework to study the US Republican Primary, using individual-level responses from a ranked-choice national poll by WPA Intelligence and FairVote in September 2023 \cite{OtisLaverty2023}. While US Presidential primaries are not held via a single national ranked-choice vote, the dynamics of state-by-state elections over time\footnote{The Republican Party selects its presidential nominee through a series of primary elections and caucuses held in each state, starting in January 2024. Voters choose delegates who represent them at the national party convention. The candidate with the majority of delegate votes at the convention becomes the party's official nominee for president.} are reminiscent of such voting, with strategic behavior, coalitions, and the effects of candidates dropping out explicitly discussed \cite{Boehm2023}:

\begin{quote}\setlength{\leftskip}{-2em}\setlength{\rightskip}{-1em}- \textit{Trump's Rivals Are Facing Loserdom. They Have a Long-Shot Option. If Haley or (less plausibly) Scott comes in second and DeSantis falls to third, the Florida governor should drop out and endorse the winner. If DeSantis wins but Haley is leading in New Hampshire, then he should offer a place on his ticket, and she should accept. Christie should then obviously drop out pre-New Hampshire and endorse the Iowa winner as well.} (New York Times, \citet{Douthat2023}) \end{quote}  \begin{quote}\setlength{\leftskip}{-2em}- \setlength{\rightskip}{-1em}\textit{Haley is solidly second in South Carolina. What will it take for her to overtake Trump in the primary? With Haley in a steady second, further growth in her SC support may depend on getting those former Tim Scott voters on board early \dots}  (The State, \citet{Bustos2023})\end{quote} %
\begin{quote}\setlength{\leftskip}{-2em}\setlength{\rightskip}{-1em}- \textit{Christie doesn’t rule out working with Haley to defeat Trump} (The Hill, \citet{Sforza2023}) \end{quote} %
FairVote thus ran a ranked-choice national poll (oversampling early primary states) to understand potential race dynamics, stating that ``examining voters’ backup choices helps us understand which candidate voters might prefer if their favorite drops out of the race'' \cite{OtisLaverty2023}. However, such analyses are limited by the ability to efficiently navigate multistage dropout dynamics and potential strategic campaign behavior; we leverage our computational framework to do so, both assuming that the poll perfectly represented the status quo and under uncertainty.

\subsection{Data and Methods}
We use the ranked-choice national poll conducted September 28-30 by WPA Intelligence and FairVote \cite{OtisLaverty2023}, with 801 respondents and 13 Republican candidates.\footnote{We received the individual-level responses via private correspondence with FairVote.} \Cref{tab:gop_primary} contains the results of a ranked-choice vote run on the responses; notably, Haley has the \textit{fifth} most {first}-choice votes (after Trump, DeSantis, Ramaswamy, Christie), but emerges as \textit{second} in the ranked-choice voting, both predicting the status quo as of early February 2024 and reflecting the complex dynamics of such races. Given the number of candidates in the poll, a strategic analysis -- studying the effects of vote additions for any subset of candidates -- would naively be prohibitive. We thus use our algorithmic approach from \Cref{sec: cut_search_space} to reduce the space of candidates and structures (Theorem \ref{thm: remove_irrelevant_candidates}), and then study strategic campaigning behavior (Theorem \ref{thm: poly_efforts_B}). We first assume that the poll perfectly represents the state of the race; then, we study strategic robustness by injecting uncertainty, by bootstrapping (sampling with replacement the individual respondents).  For writing brevity, we often refer to candidates by the first letter of their last name, with the mapping given in \Cref{tab:gop_primary}. %

The algorithmic framework developed in the previous section enables the analysis. With the full polling data, our candidate elimination algorithm (which takes 20 seconds to run on a modern laptop) is used to remove the bottom 8 candidates, who cannot win under reasonable budgets of additional votes, leaving us with $n = 5$ candidates and $m = 100$ unique ballots. Identifying optimal strategies for each remaining candidate to win or place in second requires $n! \times n$ operations each with complexity $O(mn)$; this optimal strategy identification requires around $1$ second for each candidate. Our algorithm removes at least 7 irrelevant candidates to find optimal strategies in $98\%$ of the bootstrap samples.  Without such elimination and extrapolating runtime with $n = 13$ candidates and $m = 801$ unique ballots, this analysis would require around $t = 1186$ days, and $1000 t$ days for Bootstrapping. Given that removing irrelevant candidates is sufficient for tractability, we do not leverage the algorithm from \Cref{thm: substtheorem} to further remove data-dependent infeasible sequences, though do use Proposition \ref{cor: sequences} for the number of feasible sequences. The details of all our computations are provided in \Cref{app: casestudies}.

\subsection{Results}

\begin{table}[tb]
\small
\begin{tabular}{|c|c|c|}
\hline
Candidate & Strategic additions necessary ($\%$) & Status quo head-to-head against Trump \\ \hline
[T]rump             & $0\%$ & -                 \\ \hline
[D]eSantis             & $+12.87\%$  to D  &    $43.75\%$          \\ \hline
[R]amaswamy             & $+17.75\%$  to R &   $41.27\%$            \\ \hline
[H]aley             & $+24.8\%$ to H  &   $37.74\%$       \\ \hline
[C]hristie             & $+36.35\%$  to C  & $31.87\%$       \\ \hline
\end{tabular}
\normalsize
\caption{Minimum percentage of additional votes needed to win the RCV vote. We find that the most efficient path for each candidate to win overall would be to add votes to themselves. Note the complex interaction of the full rankings and the RCV process -- Haley has the fifth-highest first-choice votes and would end up as second in the status quo, but needs the fourth-highest number of new first-place votes to end up as first. This follows from the fact that Haley's significant number of voters rank DeSantis and Ramaswamy before Trump, but not vice versa. 
}
\label{tab: votestowin}
\end{table}

\subsubsection{Strategies under perfect information}
We first assume that the poll results are correct, and ask: what strategies do each candidate need to follow (in terms of \textit{adding votes} to themselves or other candidates) to either win or place in the top two? 
\paragraph{Strategies to win the election}
\Cref{tab: votestowin} contains the optimal vote addition strategies that the top five candidates would need to follow to win the election, along with their status quo head-to-head results against Trump. We find that all the optimal strategies are \textit{selfish}, i.e., the optimal way to win overall is to make sure that new voters place them first. Notably, the minimum number of additional notes does not correspond to the current RCV output (which has Haley second, as is occurring in the actual primary): to win, Haley needs more new first-place votes than Ramaswamy, who needs more than DeSantis, while the RCV outputs the reverse order. This finding is reflected in the media, which noted that Haley would struggle to beat Trump as many of the other candidates' voters would flock to Trump as candidates drop out \cite{DawseyArnsdorfReston2024}.

\begin{table}[h]
\small
\begin{tabular}{|c|c|l|c|}
\hline
Candidate       & Percent additions & Strategies                        & Strategy type                 \\ \hline
{[}R{]}amaswamy & $1.25\%$          & $+1.25\%$ to C                    & Altruistic to losers          \\ \hline
{[}D{]}eSantis  & $2.87\%$          & $+2\%$ to C and $+0.87\%$ to D    & Altruistic to losers, Selfish \\ \hline
{[}C{]}hristie  & $4.62\%$          & $+3.75\%$ to C and $+0.87\%$ to D & Altruistic to losers, Selfish \\ \hline
\end{tabular}
\normalsize
\caption{Percent votes required to be top 2, with up to $5\%$ additions allowed. In each case, the given candidate wins alongside Trump, who remains in first place. All strategies involve ranking only one candidate, and would not benefit from voters who also add other candidates to their ranked list; for D, we note that adding $+2\%$ ranked ballots [C,D] doesn't remove the need for $+0.87\%$ [D], as the [D] votes are required so that C gets eliminated after H, and not D.  Adding [C,D] doesn't increase D's votes until C is eliminated. 
Similarly, C adds $+0.87\%$ [D] so that R is eliminated post H instead of D, as D's removal makes R stronger than E. %
}
\label{tab: top2strategies_5}
\end{table}

\paragraph{Being in the top two}
To demonstrate the flexibility and power of our computational approach, we now ask the question: what votes would a candidate seek to add to place in \textit{the top two} final candidates?\footnote{Practically, being eliminated later may also historically correspond to being nominated in a future election or otherwise receiving a cabinet appointment, or winning in the event of a late-breaking surprise.} Of course, any candidate can reach the top two if they were able to add an unlimited number of votes, and so here we limit candidates to adding up to $5\%$ additional votes. (In \Cref{app: casestudies}, we also present results for up to $3\%$ and $4\%$). With this limit, we find that only the five leading candidates have feasible strategies to reach the top two -- in addition to the existing top two of Trump and Haley: Ramaswamy, Christie, and DeSantis. For each of the latter three candidates, \Cref{tab: top2strategies_5} shows the strategies they could follow to reach the top two. Notably, these strategies are substantively different (and cheaper) than the ones required to finish first.

In particular, the optimal strategies are not necessarily to add votes to oneself. For each of DeSantis, Ramaswamy, and Christie, the optimal path to beating Haley is to add votes primarily to Christie -- so that he can beat Haley in Round 9, before eventually (for DeSantis or Ramaswamy) losing. This strategy emerges because many of Christie's votes place Haley second, but the reverse is not true. Thus, helping Christie beat Haley is far more efficient for DeSantis and Ramaswamy than is adding first-place votes to themselves: they would need to add $6.63\%$ and $4.25\%$  additional self-votes\footnote{ [D]eSantis gets eliminated in Round 10, where we need 0.6$\%$ to D so that [R]amaswamy gets eliminated. After R's transfers, the gap between [H]aley and D becomes 6.63$\%$, implying total $6.63\%$ selfish additions make D beat H to reach top 2. Similarly, for R, $4.25\%$ self-votes suffice.}, respectively, instead of $2.87\%$ and $1.25\%$  additional votes needed when primarily supporting Christie. Surprisingly, these optimal strategies just require that the new voters place Christie first, and do not require, e.g. for DeSantis, that they place DeSantis second in their rankings.

~\\\noindent Together, these findings represent the power of our approach: we can efficiently search the space of strategies to find optimal strategies for each candidate, depending on their goal (of either winning or coming in the top two). Notably, for the goal of reaching the top two, we discover that it would be more efficient for some candidates to primarily support other candidates than to try to increase their own vote totals. %

\begin{table}[tb]
\centering
\small
\begin{tabular}{|l|c|c|c|c|}
\hline
Strategy (\%) & H & R & D & C \\ \hline
Original (No Strategy) & 68.8 & 24.4 & 5.3 & 1.5 \\ \hline
[R]amaswamy's Strategy (+C=1.25$\%$) & 46.7 & \textbf{38.6} & 9.7 & 5.0 \\ \hline
[D]eSantis's Strategy (+C=2$\%$, +D=0.87$\%$) & 33.1 & 38.1 & \textbf{19.7} & 9.1 \\ \hline
[C]hristie's Strategy (+C=3.75$\%$, +D=0.87$\%$) & 11.5 & 37.9 & 23.8 & \textbf{26.8} \\ \hline
\end{tabular}
\normalsize
\caption{Each candidate column is the frequency of bootstrap samples in which the candidate places in the Top 2, when optimal strategies according to the full data for candidates in the row are used. With uncertainty, otherwise optimal strategies that operate by helping other candidates may be ineffective.}
\label{tab:strategy_comparison}
\end{table}

\begin{table}[tb]
\small
\begin{tabular}{|l|l|l|l|l|l|l|}
\hline
Candidate & H & R & D & C & P & Sc \\ \hline
Top 2 frequency in bootstrap samples (\%) & 69.69 & 23.67 & 5.20 & 1.43 & - & - \\ \hline
Top 2 frequency under 5\% strategic additions (\%) & 99.49 & 91.63 & 73.67 & 48.37 & 2.45 & 0.20 \\ \hline\hline
Average additions in strategy (\%) & 1.20 & 2.08 & 2.74 & 3.18 & 4.30 & 4.06 \\ \hline
\end{tabular}
\normalsize
\caption{Under bootstrap polling samples, how often each candidate appears in the top 2 either without strategic additions or under optimal (for that sample) additions of up to $5\%$. \textit{P} is Pence, \textit{Sc} is Scott. }
\label{tab: bootstrap_summary_table}
\end{table}

\begin{table}[tbh]
\centering
\small
\setlength{\tabcolsep}{5pt} %
\renewcommand{\arraystretch}{1.2} %
\begin{tabular}{|c|c||c|c||c|c||c|c||c|c||c|c|}
\hline
\multicolumn{2}{|c||}{Haley} & \multicolumn{2}{c||}{Ramaswamy} & \multicolumn{2}{c||}{DeSantis} & \multicolumn{2}{c||}{Christie} & \multicolumn{2}{c||}{Pence} & \multicolumn{2}{c|}{Scott} \\
\hline
Type & \SI{}{\percent} & Type & \SI{}{\percent} & Type & \SI{}{\percent} & Type & \SI{}{\percent} & Type & \SI{}{\percent} & Type & \SI{}{\percent} \\
\hline
H   & \SI{90.41}{} & C   & \SI{39.19}{} & DC   & \SI{50.52}{} & C   & \SI{61.09}{} & P   & \SI{62.50}{} & Sc  & \SI{100.0}{} \\
HR  & \SI{3.42}{}  & R   & \SI{28.38}{} & D   & \SI{24.59}{} & DC  & \SI{27.83}{} & CP  & \SI{29.17}{} &     &              \\
HD  & \SI{2.40}{}  & RC  & \SI{20.57}{} & C   & \SI{16.99}{} & CR  & \SI{5.65}{}  & RP  & \SI{4.17}{}  &     &              \\
HC  & \SI{2.40}{}  & DC  & \SI{6.16}{}  & RDC & \SI{3.13}{} & RDC & \SI{2.17}{}  & DP  & \SI{4.17}{}  &     &              \\
D   & \SI{1.37}{}  & RDC & \SI{4.05}{}  & RC  & \SI{1.79}{} & D   & \SI{1.74}{}  &     &              &     &              \\
    &              & RD  & \SI{0.75}{}  & HD  & \SI{1.79}{} & HC  & \SI{1.09}{}  &     &              &     &              \\
    &              & D   & \SI{0.60}{}  & RD  & \SI{0.45}{} & CP  & \SI{0.22}{}  &     &              &     &              \\
    &              & HR  & \SI{0.30}{}  & DCP & \SI{0.15}{} & HRC & \SI{0.22}{}  &     &              &     &              \\
    &              &     &             & RDP & \SI{0.15}{} &     &              &     &              &     &              \\
\hline
\end{tabular}
\normalsize
\caption{Strategy categorization under uncertainty for each pair of candidates that can reach the top, given up to 5$\%$ vote additions. Each column, summing up to $100\%$, shows the distribution of optimal strategies that each candidate applies, whenever they are not in the top two in the bootstrap sample without strategic additions. To reduce the number of listed strategies, when a strategy is denoted XYZ (like HR, DCP, etc), the actual optimal additions are to separately add ballots for X, Y, or Z, and/or add ranked ballots of some combination of those candidates, e.g., [Y, X], [X, Y, Z].}
\label{tab: combination_percentages}
\end{table}

\subsubsection{Robustness under uncertainty} One substantial limitation of the above analysis is that it -- incorrectly -- assumes that the poll responses perfectly reflect the state of the race. Under any uncertainty or randomness, a strategy leading to a close victory can instead lead to a close loss -- and the metaphor between the RCV poll and the actual nomination process likely breaks down.\footnote{For example, the data in \Cref{tab:gop_primary} shows that the margin with which the eliminations occur is slim in some rounds; in rounds 1, 3, and 5, it is of the order of 0.1-0.2$\%$. The strategies from \Cref{tab: top2strategies_5} additionally show how smaller deviations may lead to significant updates to the social choice order.} Here, we analyze the robustness of the above strategies under such uncertainty. We ask: How stable are these strategies if the underlying poll data is noisy? Under imperfect information, do candidates know how to pick their coalition partners? How costly are strategies, under uncertainty? We focus on the case where each candidate seeks to place in the top two, informed by our analysis above that it would require a large number of new votes for any candidate to place first over Trump.

To characterize uncertainty, we use the method of bootstrapping, resampling with replacement the poll respondents.\footnote{This approach does not fully capture error. In non-ranked-choice polling, \citet{shirani2018disentangling} find that the polling margin of error often understates error -- e.g., caused by sampling biases and changes of opinions over time -- by a factor of two. Thus, real-life strategies would likely need to incorporate even more uncertainty than in our case study analysis.} For each resample, we (a) calculate the efficacy of the optimal strategies found above, and (b) find the optimal strategies assuming \textit{that sample} of respondents is the ground truth. We note that computational efficiency is key to doing this procedure, as we need to analyze each bootstrapped sample of voters. Notably, our algorithm removes at least 7 irrelevant candidates to find optimal strategies in $98\%$ of the samples for up to $5\%$ strategic additions, with $99.7\%$ and $99.9\%$ efficacy for $4\%$ and $3\%$, respectively.

\textbf{(a) Efficacy of optimal strategies calculated above.} We consider the optimal strategies for each candidate from Table \ref{tab: top2strategies_5} (for the original data) and apply those to each bootstrap sample. \Cref{tab:strategy_comparison} contains the resulting outcomes, in terms of the percent of samples in which a given candidate is in the top two. Table \ref{tab: top2strategies_5} shows that the optimal strategy given the full data for each of candidates R, D, and C heavily relies on adding votes to C. Under imperfect information, thus, an application of these strategies leads to an increased frequency of wins for each three of them -- but does not necessarily lead to a win for the candidate for whom that strategy was originally optimal. In other words, optimal strategies, especially those that add votes to other candidates, may be fragile.

\textbf{(b) Distribution of optimal strategies under each sample.} Now, we analyze robustness using an alternative manner: how similar are the optimal strategies calculated for each bootstrap sample of the respondents? Low strategy variance would suggest that candidates would follow a similar strategy up to small perturbations of the polling data. \Cref{tab: bootstrap_summary_table,tab: combination_percentages} show the analogue of \Cref{tab: top2strategies_5} under uncertainty. The former shows the six candidates after \textit{T} who either reach the top two directly in a sample or can do so after a strategic addition of up to 5\% of votes. The latter further categorizes the strategic additions for each candidate when they aren't part of the top 2.  %

Qualitatively, we find that optimal strategies that primarily add votes to oneself (e.g., for \textit{H}, \textit{C}, \textit{P}, or \textit{Sc}) are consistent across samples. Conversely, strategies that primarily operate by aiding other candidates (such as \textit{R} or \textit{D} adding votes for \textit{C}) are not consistently optimal -- further speaking to the fragility of complex strategies. On the other hand, purely selfish strategies may be more costly (e.g., $6.63\%$  selfish vs $2.87\%$ coalition for \textit{D} as in the case with perfect information). %

~\\\noindent Putting things together, the results speak to the difficulty of designing effective campaign strategies with RCV voting (or metaphorically similar processes such as state-by-state primaries over time, with candidates dropping out): purely selfish strategies of adding votes for oneself are costly (require substantial additions), but other actions such as supporting other candidates may backfire under imperfect information. On the other hand, we find that optimal addition strategies are often fairly simple: they do not require voters to fill \textit{long} ballots in a specific order to lead to a specific elimination chain, but can work even if the additional ballots only list a few candidates. We also provide case study-oriented insights: (a) H is a \textit{robust} second-place choice unless other candidates can support C as a spoiler but would have difficulty winning even after multi-round transfers; (b) it would guide candidates for \textit{how} to strategize, if they find it tractable: e.g., by informing D and R approximately how many votes they must add to C to help him serve as a spoiler.

\section{Theoretical Analysis of Strategies}
\label{sec:theory}

The results in Section \ref{sec: casestudies} bring us to our final exploration: How do we understand the empirical patterns? Can the non-selfish strategies lead to, perhaps, any stable coalitions?
To that end, we theoretically analyze the vote-addition strategies as well as extend the analysis to other persuasion efforts, such as cross-endorsements or coalitions in practice. The quest is clear: It is to understand the types of actions that would benefit the candidates \textit{at all}, in general; and how this would change under uncertainty. We subsequently define what it would mean to `benefit' in this sense.

\begin{definition}
A candidate or a group of candidates \textit{benefits} via an action, if it increases the number of their active votes in any round, i.e., their first-choice and received transfer votes, before they get eliminated or win a round. In other words, receiving `benefit' implies that the array of votes they have in each round (while in the contest) would strictly increase at some position. 
\end{definition} 
Benefiting is a weaker notion than `winning' -- if a strategy cannot benefit a candidate as defined, then it cannot help them win either. We study this notion in two settings: when candidates have perfect information about the \textit{current} voter ballots, or when this information is uncertain. With uncertain information, we distinguish between strategies that may work \textit{ex-ante} versus \textit{ex-post}. In practice, candidates have uncertain information about the state of the race such as via polling data. 

We consider two action types: (a) \textit{strategic voting} -- \textit{changing} existing votes, e.g., keep the voter set the same but \textit{persuades} them to change their votes; and (b) \textit{strategic additions} -- \textit{adding} votes, e.g., \textit{turn out} new voters with a specific ballot profile. Note that adding votes changes the Droop quota but changing votes does not. Naively, such actions can be arbitrarily complex, given the per-round transfer behavior of votes -- e.g., adding votes ranking all $n$ candidates in a particular order could have vastly different outcome implications than adding the same votes but with two candidates' positions swapped. As another example, it may be ``cheaper'' (in the number of votes needed to change or add) for a candidate to prevent another candidate's elimination and help them serve as a spoiler for a true competitor, instead of increasing their own votes. We do not consider \textit{equilibria}, but rather the simpler question of what actions could work even holding all other votes fixed.\footnote{More generally, candidates do not directly control voter ballots or turnout -- they control their own actions such as campaigning, advertisements, and turnout infrastructure \cite{enos2018aggregate}. Thus, implementing strategies that work in theory may not be practical. It is thus important to find strategies that both require few changes and are simple, that do not involve voters implementing a long ranked list exactly.}

\subsection{Strategic behavior under perfect information} \label{sec: perfect_info_strategies}

Suppose a candidate had perfect knowledge of the status quo voter preferences and who would turn out. How should this candidate seek to change preferences and turnout? We first analyze strategic voting (changing preferences) and show that it isn't practically useful from the candidates' point of view. Then, we study strategic additions (changing turnout). The two are further connected: the implications of strategic additions may be translated to strategic voting as well: if the addition of $X$ votes changes the structure, it may also be achieved by altering fewer preferences of current voters.

\subsubsection{Strategic voting -- altering the existing preferences: }
Can a candidate (or coalition of candidates) improve their outcomes by persuading their own voters to change their preference lists? 

\begin{assumption}
    A candidate $c$ can only affect a voter's ranking \textit{after} that voter places $c$. For example, if a voter currently ranks $c$ third, then $c$ cannot persuade that voter to change their top two positions, but can arbitrarily change the voter's preference starting at the third position. \label{assump:ownvoterbase}
\end{assumption}

This assumption is relevant when a candidate is directly telling their voters how they should rank other candidates, such as via endorsements. Such coalition behavior often occurs in the form of cross-endorsements in practice (e.g., multiple instances in New York City's RCV election \cite{FairVote2024}). Under this assumption, we observe that no actions are self-beneficial, either for an individual or a coalition.

\begin{proposition}
\label{prop: no_self_benefit}
Under Assumption \ref{assump:ownvoterbase} and perfect information, individual candidates cannot benefit via strategic voting. Moreover, strategic voting by a coalition of candidates can benefit some candidates in the coalition, but cannot benefit all candidates. %
\end{proposition}

\begin{proof}
The argument follows from the fact that a candidate's influence on ranked choices after their position is accessible only if they are either eliminated or have already won. In both scenarios, the candidate cannot use this influence to their own advantage. Now, consider a coalition of $m$ candidates. We demonstrate that at least one candidate in the coalition does not gain from this arrangement. Consider the first coalition candidate $c_1$ who is either eliminated or wins among the coalition of $m$ candidates. For either result, there is no benefit or vote transfer from the coalition as the ballots containing strategically altered preferences following other coalition members' ranks are still with their prior choices.
Therefore, in either case, $c_1$ does not benefit from being part of the coalition. However, after persuading voters to vote in a bloc of coalition members, the last candidates remaining in the contest could see the benefits of coalitional strategic voting, with possible translations to the wins.
\end{proof}

\Cref{prop: no_self_benefit} is an impossibility result, under perfect information: a group of candidates cannot all benefit from a coalition, where they collectively employ strategic voting to tell their voters to rank the other candidates after themselves -- even if the number of winners $k$ in the election is greater than the coalition size.  The result suggests that the justification for such coalitions in practice must come from uncertainty on existing voter preferences (as we'll see in \Cref{lem: coalitions_imperfect}, coalitions can benefit all candidates ex-ante but not ex-post under imperfect information).

\subsubsection{Strategic additions -- adding new ranked votes}

\begin{table}[tb]
\centering
\begin{tabular}{l|p{10cm}}
\textbf{Strategy name} & \textbf{Description} \\ \hline
Selfish & Add ballots that rank yourself. (Pure selfish: rank only yourself) \\ \hline
Altruistic to winner(s) & Add ballots that rank candidates who will win %
and other candidates (including potentially oneself), to get their transfer votes.
\\ \hline
Altruistic to loser(s) & Add ballots that rank candidates who still lose, %
and other candidates (including potentially oneself),
so that they may serve as spoilers to other candidates. 
\end{tabular}
\caption{Example strategies for voting additions. These examples are non-exhaustive, and some strategies may be a mixture of several strategies. %
}
\label{tab: coalition_categories}
\end{table}

Given the ineffective nature of actions that only affect one's own existing voters, we now study what happens if the candidate's action space is expanded: they can add votes, and the new voters can have arbitrary rank lists. %
Naively, many strategies could be optimal for a candidate or a coalition -- \Cref{tab: coalition_categories} categorizes such strategy types.

Our next result classifies the optimal addition strategies. We show that it is largely sufficient to consider ballot additions of a single category type: ballots that first rank candidates that will still be round-losers, and then a single rank of a candidate who will be a round (and thus election) winner. There are always optimal ballot structures of this form in the case of single-winner RCV.

In the case of multi-winner RCV, there is a ``corner case'' in which adding another ballot type may instead be optimal, corresponding to the `Altruistic to winner(s)' strategy. We define the case as Case (A): There is a candidate who receives \textit{low top-choice support}- leading to elimination in the status quo, but \textit{high later-choice support}- leading to a sure win if some vote additions prevent elimination in early rounds. Formally, the election instance contains at least one candidate who needs $X$ minimal additional votes that prevent elimination in the first $r$ round(s) and later wins a round ($\geq r+1$) with {only} the transfer votes from out-of-contest candidates, i.e., without needing any added votes in that round. %

\begin{restatable}{theorem}{ranklosingc}\label{thm: rank_losing_C}
Given voter data and a target set of election winners, the optimal strategies are completely characterized based on Case (A) occurrence:
\begin{itemize}
    \item In the absence of Case (A), for the optimal vote-addition strategies it suffices to add ballots of the form $[L,L,\dots,L,W]$ or $[L,L,\dots,L,L]$, including  $[L]$ and $[W]$, i.e., only losing candidates are ranked until the last choice on the ballot. It is strictly beneficial for the strategist to be the last candidate on such ballots, i.e., making these `selfish' plus `altruistic to losers'.
    \item In Case (A), additionally the strategy `Altruistic to winner(s)' of the form $[\dots, W, \dots]$ may be optimal, i.e. containing at least one winning candidate ranked at non-terminal positions.
\end{itemize}%
An optimal strategy may add multiple types of ballots, all of the specified forms. Strategies may not be unique: we may always add more preferences (if any) at the end of these ballots and maintain the same result (thus, we say that it suffices to add ballots of this form; adding irrelevant candidates to the end of ballots may break the form). %
\end{restatable}  

\begin{proof}[Proof Sketch]
We show that only under Case (A), candidates can benefit from adding ballots of the form $[\dots, W,\dots]$. Intuitively, if a candidate has a larger later-choice appeal, their active votes may increase over time, potentially allowing transfers of the added votes post-winning. But, if they need extra votes to win a round, i.e. reach the Droop quota, they do not generate surplus, hence no one can use these ballots post they win.  In the absence of Case (A), we demonstrate that adding votes to losing candidates until the second-last choice is a more efficient strategy than other methods of vote allocation. This follows from the fact that surplus transfer after a round-win is weighed down, whereas transfer post an elimination is not. See \Cref{app: sec3details} for details. %
\end{proof}

The theorem establishes that except under the rare instance of Case (A),
 `selfish' and `altruistic to loser' strategies strictly dominate the `altruistic to winner' strategy. Thus, it largely suffices to consider the combinations of `selfish' and `altruistic to loser' as optimal strategies for a candidate or a group of candidates. A candidate may seek to add voters who support other candidates (who will eventually lose) over themselves. Section \ref{sec: cut_search_space} establishes a computational framework to compute exact optimal strategies given an election instance.

\Cref{thm: rank_losing_C} qualitatively extends \Cref{prop: no_self_benefit} to the vote addition strategies: with perfect information, only under the instance of Case (A), candidates may seek coalitions that can potentially benefit all involved. Otherwise, no group of candidates can seek to collectively add optimal voter ballots that benefit all coalition members, since optimal ballots are of the form that ranks losing candidates first and then a single winner. (Note that it is not a direct extension; there are of course ballot additions that benefit all candidates in the coalition, but they may not be optimal for any given candidate). 
It further conceptually establishes that in RCV or STV elections, there is an incentive to seek endorsements from candidates who will themselves lose the race.

\subsection{Strategy dynamics under imperfect information }\label{sec: imperfect_info_strategies}
The above analysis is under a strong assumption: that the candidates can strategize perfectly knowing the ``status quo'' of the race, i.e., the full distribution of votes from which they can calculate strategies. In practice, this knowledge has high uncertainty -- polling data is imperfect and opinions may change over time even without actions. Especially given the intricate round-specific behavior of transfers, optimal strategies assuming one status quo vote set may be ineffective or counterproductive under a ``nearby set.'' We now study optimal strategies under such uncertainty. Now, assume that the current vote set is drawn from a known distribution $\mathcal{D}$, and the actions change or add votes as above.

We first study strategic voting (changing existing votes) under uncertainty. \Cref{prop: no_self_benefit} immediately implies that strategic voting under Assumption \ref{assump:ownvoterbase} (that one can only influence one's own voters) cannot benefit all coalition candidates \textit{ex-post}, i.e., under any specific realization of existing votes. However, the following result establishes that it can benefit all members \textit{ex-ante}, e.g., increase the probability that they each win the election. %

\begin{proposition}
\label{lem: coalitions_imperfect}
Under uncertainty and Assumption \ref{assump:ownvoterbase}, strategic voting among a coalition of candidates does not ex-post benefit all candidates, although it can ex-post increase the number of winners from the coalition. However, there exist distributions $\mathcal{D}$ such that every candidate can ex-ante benefit from joining a coalition.  %
\end{proposition} 
\begin{proof}
Consider a strategy that strictly increases the number of voters voting within the coalition bloc (voters who rank all candidates in the coalition before candidates outside it), without changing their prior preferences within the coalition. %
The formation of a bloc coalition might lead to more wins from the coalition as fewer votes get exhausted or go outside the coalition as eliminations occur, benefiting the last remaining candidate(s) from the coalition. Accordingly, some of the benefits may be translated to the number of wins.\footnote{Proposition 1 in \cite{garg2022combatting} gives bounds on the number of wins under two-party bloc voting in multi-winner STV.} 
Being part of the coalition then strictly increases the voter base of candidates, either before or after their elimination. Given any structure as ground truth, as voters are not updating their preference orders within the bloc, no candidate within the coalition can see a drop in their active voter base post-joining the coalition. 
Then under uncertainty of election structure, all candidates ex-ante benefit via the strategy. 
\end{proof}

The result justifies coalitions that cross-endorse each other, with each candidate asking their supporters to rank the other candidates immediately after them: it may be mutually beneficial to do so in probability, in the sense that it increases the probability that each candidate wins the election under uncertainty on the status quo state of the race.

Finally, we study the robustness of optimal vote addition strategies as derived in \Cref{thm: rank_losing_C}. The following result indicates that \textit{purely selfish} strategies that rank the target winner first (in additional votes) are robust (\textit{ex-post} beneficial under any status quo), while other strategies such as \textit{altruistic to losers} (supporting other candidates so they serve as spoilers before losing) are not.

\begin{proposition}\label{lem:rank_losing_C_imperfect}
    Even under uncertainty, the purely `selfish' strategy of adding first-place votes to oneself is always ex-post beneficial, while any other strategy may be ex-post disadvantageous even if ex-ante beneficial.
\end{proposition}
\begin{proof}
    The selfish strategy is always beneficial as it increases the active voters of the candidate in every round -- even under any true status quo votes, this cannot eliminate the candidate earlier than its position without additions \footnote{The upper monotonicity paradox states that asking voters to shift your ranking up may worsen your position \cite{graham2023examination}. However, adding new `selfish' votes is always beneficial.}. However, strategies that add votes to other candidates first may become strictly disadvantageous under imperfect information, if the strategist gets eliminated early in a competition before the beneficiary, by a margin less than the added votes.
\end{proof}
This implies that although `altruistic to losers' may be cheaper than the `selfish' strategy, such strategizing under uncertainty may be costly. In STV coalitions, strategic additions become even more challenging.  If two candidates $C_1$ and $C_2$ (with $V_{C_1} > V_{C_2}$ and bloc voting in multi-winner STV)  form a coalition, they need to ensure that $C_1$ wins and transfers significant votes to $C_2$, before $C_2$ is eliminated or $k$ candidates are declared winners. The feasibility of this strategy depends on structure-specific information coming from the available voter data. In this setting under uncertainty, $C_1$ may prefer to strategically add votes such that $C_2$ loses even earlier, transferring all their votes to $C_1$. 

Putting things together, in this section we derived the types of strategies that can possibly work or are optimal, both under uncertainty and perfect information. The analysis implies that stable coalitions may be reasoned only under imperfect information, while perfect information favors both relatively costly ‘selfish’ as well as potentially backfiring ‘altruistic to loser’ strategies. 

\section{Conclusion}
Rank choice voting (RCV) is gaining popularity, and there's a growing push for its widespread implementation. However, understanding and navigating RCV can be complex, often requiring significant computational effort. This paper aims to shed light on these complexities.

We analyze the actions that benefit a candidate or group of candidates in Single Transferable Vote (STV) elections. Focusing on strategic additions, e.g. campaigning, we develop tools to find optimal strategies for candidates, making it easier to understand how close a candidate is to winning and to perform rigorous statistical analyses of voter data, even under uncertainty using methods like bootstrapping. Extending our analysis to cross-endorsements or coalitions, we find that stable coalitions can only be reasoned under imperfect information.

In our algorithmic analysis, we build a combinatorial framework for STV and establish an algorithm (polynomial-time in the number of orders of candidates) for achieving the desired STV election outcome. This framework can be extended to any deterministic elimination rules that map to a social choice order, such as post-batch elimination. Given the combinatorial explosion of possibilities precluding full efficiency and approximation guarantees, we provide polynomial-time algorithms for reducing the search space. These work with increasing efficiency as the persuasive ability, measured by the number of votes to be added, decreases. This offers a tractable `upper bound' for the general persuasion effect. Unlike the plurality rule, the election dynamics, including the notion of how close or how far (from winning or within the candidate pool) are quite opaque in RCV, and our algorithms aim to clarify these dynamics optimally. Future work could focus on finding tighter search spaces and modeling the general persuasion effect, with both theoretical and practical implications.

We demonstrate the practical implementation of our work using data from a 2024 Republican Primary RCV poll. Calculating optimal strategies naively is computationally prohibitive due to the large number of candidates and the combinatorial nature of RCV elimination. Our approach of eliminating irrelevant candidates (in our case study, the bottom six candidates after Scott were irrelevant in $98\%$ of bootstrap samples with a budget of only $5\%$ additional votes, even as spoilers) and then searching through the remaining structures proves tractable. Our analysis shows that optimal strategies for candidates often involve supporting spoilers, but these strategies may not be robust under uncertainty. Overall, our analysis offers a clearer picture of RCV poll data, exactly telling how far each candidate is from reaching the top, and closely informing strategy dynamics under uncertainly.

Our theoretical analysis further explores manipulation strategies, focusing on cases where a candidate can influence ballot ranks below their own and addition strategies, where arbitrarily ranked votes may be added. This is motivated by real-world instances, such as in NYC, where candidates endorse each other and ask their voters to rank the others after themselves. We find that these endorsements are justified by uncertainty; with perfect information, selfish and spoiler-inducing strategies are optimal. Given the risky nature of enabling spoilers, this suggests that RCV is not easily manipulable in practice. Future work could explore other manipulations and model their associated costs more robustly.

Overall, our work supports the use of STV and RCV in practice: while supporting spoilers and otherwise behaving strategically is theoretically possible, doing so under uncertainty of race conditions or the exact consequences of one's electoral behavior is risky -- both theoretically and empirically. 

\newpage
\bibliographystyle{plainnat}
\bibliography{biblio}

\newpage
\appendix

\section{Introduction: Details}\label{app: RCV_examples}
\begin{exmp} \label{ex: campaign}
    Consider candidates $\mathcal{C} = \{A,B,C,D\}$, contesting in a 2-member district, with a total voter population of 100 (Droop quota = $34$). The original election profile is given in Table \ref{tab: ex_1_original}, leading to the social choice order $f_{original} = A>C>B>D$. 
    \begin{table}[ht]
    \centering
\begin{tabular}{|c|c|c|c|c|c|}
\hline
  & Round 1 & transfer & Round 2 & transfer & Round 3 \\ \hline
A & 30      & 10        & 40 (W)  & -        & -       \\ \hline
B & 25      & 3        & 28      & 1        & 29      \\ \hline
C & 23      & 9        & 32      & 5        & 37 (W)  \\ \hline
D & 22 (L)  & -        & -       & -        & -       \\ \hline
\end{tabular}
\caption{Original Election profile}
\label{tab: ex_1_original}
\end{table} 
\normalsize
Now, candidate $B$ strategizes by adding two votes to candidate $D$, i.e. adds two  [$D$] ballots. Now, the election profile is as given in Table \ref{tab: ex_1_campaign}, leading to the updated order $f_{new} = B>A>D>C$.
\begin{table}[ht]

\centering
\begin{tabular}{|c|c|c|c|c|c|}
\hline
  & Round 1 & transfer & Round 2 & transfer & Round 3 \\ \hline
A & 30      & 4        & 34      & 3        & 37 (W)  \\ \hline
B & 25      & 15       & 40 (W)  & -        & -       \\ \hline
C & 23 (L)  & -        & -       & -        & -       \\ \hline
D & 24      & 4        & 28      & 3        & 31      \\ \hline
\end{tabular}
\caption{Election profile after smart campaigning.}
\label{tab: ex_1_campaign}
\end{table}
\normalsize
Observe that if $B$ wishes to win by adding self-votes, it needs to add 9, 9, or 6 votes to win in rounds 1, 2, and 3 respectively. Using the surplus transfer effectively, $B$ can be in the winning set by adding just 2 extra votes. This works as $B$ is able to gauge that $C$'s major voter base also supports $B$. If $C$ loses before 2 candidates secure win, $B$ is assured to gain the support. Note that this information is not obvious from the original election profile but comes from the understanding of the entire election data, which needs to be navigated carefully and extensively.
\end{exmp}
\begin{exmp}\label{ex: campaign_flip}
    Again consider candidates $\mathcal{C} = \{A,B,C,D\}$, contesting in a 2-member district. The original election profile is given in Table \ref{tab: ex_2_original}, leading to the social choice order $f_{original} = A>B>C>D$. 
    \begin{table}[ht]
    \centering
    
\begin{tabular}{|c|c|c|c|c|c|}
\hline
  & Round 1 & transfer & Round 2 & transfer & Round 3 \\ \hline
A & 23      & 13       & 36 (W)  & -        & -       \\ \hline
B & 25      & 9        & 34      & 2        & 36  (W)    \\ \hline
C & 30      & 0        & 30      & 0        & 30  \\ \hline
D & 22 (L)  & -        & -       & -        & -       \\ \hline
\end{tabular}
\caption{Original Election profile}
\label{tab: ex_2_original}
\end{table} 
\normalsize
Now, one first-choice vote for A changes their vote to a first-choice vote for D. The election profile is as given in Table \ref{tab: ex_2_flip}, leading to the updated order $f_{new} = D>C>B>A$.

\begin{table}[ht]
\centering

\begin{tabular}{|c|c|c|c|c|c|}
\hline
  & Round 1 & transfer & Round 2 & transfer & Round 3 \\ \hline
A & 22   (L)   & -       & -  & -        & -       \\ \hline
B & 25      & 5        & 30      & 2        & 32      \\ \hline
C & 30      & 0        & 30      & 4        & 34 (W)  \\ \hline
D & 23  & 17       & 40 (W)       & -        & -       \\ \hline
\end{tabular}
\caption{Election profile after one vote is updated}
\label{tab: ex_2_flip}
\end{table} 
\normalsize
Here, the change is the smallest possible, i.e., one vote in 100. Another interesting point is that here, A and D voters are largely voting in blocs and support both A and D in their first two choices.
\end{exmp}
\section{Model: Details} \label{app: protocolproof}

\subsection{Social welfare function}
We show that the mechanism defined in \ref{sec: model} is equivalent to standard definitions of STV, referred to as STV protocol here. Thus, the mechanism forms a well-defined social welfare function for STV (Arrowian definition \cite{arrow2012social}) as it produces a collective social choice order, aggregating social preferences.
\begin{restatable}{theorem}{stvprotocol} \label{thm: stvprotocol}
For any $n$ candidates contesting for $k$ seats, function $\mathcal{F}_{STV}$  is a well-defined social welfare function for STV.
\end{restatable}
\begin{proof}
The Single Transferable Vote (STV) protocol and the mechanism both result in two exclusive groups: winners ($k$ candidates) and losers ($n-k$ candidates). The mechanism also creates an order ($f$) of all candidates, with the top $k$ being winners. We need to show that both systems select the same winners. In both systems, a single candidate is eliminated in a round only when there are no candidates with votes greater than $Q$. Evident from both definitions, this candidate becomes part of the winning set only when $n-k$ candidates are already eliminated. Then, the only difference is that the STV protocol can have multiple winners in a round if they all exceed $Q$ votes, while the mechanism processes one candidate per round. We now illustrate that if multiple candidates have over $Q$ votes in a round, they all become part of the top $k$ candidates in the social choice order $f$.

We start with $T$ votes total. 
Suppose two candidates $C_1$ and $C_2$ have $v_1^r \geq v_2^r \geq Q$ votes in round $r$. $C_1$ wins in the mechanism, securing the highest available position in order $f$. Future surplus vote transfers will continue to keep $v_2^{r'}\geq Q$, as $v_2^{r'} \geq v_2^r$ for all rounds $r'>r$. Therefore, no eliminations will happen by the mechanism until $C_2$ wins, since we always have a candidate with more than $Q$ votes in each round until then. This is smoothly generalized 3 or more candidates securing more than Q votes in a round.

Observe that an elimination round keeps the total number of votes the same or less, while a winning round removes at least $T/(k+1)$ votes (as the surplus transfer is made only over $Q$ votes). More votes are removed from the election if some surplus votes can't be transferred due to incomplete ballots.  %
In any round, the maximum number of votes is $T$, and hence, at most $k$ candidates can have more than $Q$ votes. Each winning round removes at least $T/(k+1)$ votes and 1 seat from the contest. Thus, at most 1 candidate can have more than $Q$ votes when only one position remains to be filled (amongst the top $k$). 

Lastly, it's clear that the mechanism mirrors the STV rule for surplus transfer, resulting in the same vote counts after each round. This confirms that a candidate that is a winner in the STV protocol must be within the top $k$ candidates in the order $f$.
\end{proof}

\subsection{Generating structure-specific constraints} \label{app: algo_constraints}

In this section, we present the construction of structure-specific constraints for STV elections.  As discussed in Section \ref{sec: setup}, a structure consists of two parts: the order and the sequences, which provide information about the aggregated social choice order and round-specific details, respectively.  These structures uniquely characterize the combinatorial data associated with STV elections and enable the partition of voter data space into regions demarcated by corresponding constraints. The constraints are used to determine the STV result for any point in the voter data space and to find the thresholds that ensure the optimality of any structure.

To facilitate this process, we introduce two algorithms:  Algorithm \ref{algo: round_update} and 
\ref{algo: constraints}, that together provide the region-defining constraints for any structure.  Algorithm \ref{algo: round_update} is used to update the vote tallies in every round, by transferring votes from the last processed candidate. The vote tally stores the path of each vote that contributes to the total, and not just the number of contributing votes. It also stores the fractions applied to each unique path, in case the path is contributed by a winning candidate's surplus. If we are working with ballot data, we may simplify the third \textit{for} loop to go over all ballots that need to be transferred; this way we don't iterate over all permutations of the out-of-contest candidates.

Algorithm \ref{algo: constraints} collects the constraints that define the outcome of each round based on the given structure. The final output of Algorithm \ref{algo: constraints} is a set of region-defining constraints, represented by the aggregated variables $V^i_{a,b,...}$. By inputting the relevant data into these variables, one can evaluate the satisfiability of any given structure for a specific point in the voter data space.

\begin{algorithm} 
\caption{Round Update for STV Elections} \label{algo: round_update}
\begin{algorithmic}[1]
\REQUIRE $structure, processed\_candidates, active\_candidates, vote\_tally$
\ENSURE Updated $vote\_tally$

\STATE $last\_checked \gets$ last candidate from $processed\_candidates$

\FOR {each vote path $path$ in $last\_checked$'s $vote\_tally$}
\STATE $path\_surplus \gets$ popped out fractions of surplus votes from $path$
    \FOR {each candidate $c$ in $active\_candidates$}
        \FOR {each unique permutation $perm$ of $processed\_candidates$ not overlapping with $path$}
            \IF{$structure[last\_checked] =$ lose}
                \STATE $new\_path \gets [path, perm, c, path\_surplus]$
            \ELSE 
                \STATE $new\_path\_surplus \gets$ fractions of surplus votes from candidate $last\_checked$
                \STATE $new\_path \gets [path, perm, c, path\_surplus, new\_path\_surplus]$
            \ENDIF
            \STATE Add $new\_path$ to $c$'s vote paths in $vote\_tally$
        \ENDFOR
    \ENDFOR
\ENDFOR
\RETURN $vote\_tally$
\end{algorithmic}
\end{algorithm}

\begin{algorithm} 
\caption{Structure-Specific Constraints for STV}\label{algo: constraints}
\begin{algorithmic}[1]
\REQUIRE $structure, set\_candidates, vote\_tally, droop\_quota$
\ENSURE Constraints for vote tallies
\STATE $processed \gets$ Empty list \COMMENT{List of processed candidates}
\STATE $active \gets$ Copy of $set\_candidates$ \COMMENT{List of active candidates}
\STATE $constraints \gets$ Empty list \COMMENT{List of constraints}

\FOR {each candidate $C$ in sequence[$structure$] }
    \STATE $vote\_tally \gets$ Round Update for STV Elections [$structure, processed, active, vote\_tally$]
    
    \IF{$structure[C] =$ lose}
        \STATE Append ``$droop\_quota \geq vote\_tally[active] \geq vote\_tally[C]$" to $constraints$
    \ELSE
        \STATE Append ``$vote\_tally[C] \geq vote\_tally[active], droop\_quota$" to $constraints$
    \ENDIF
    
    \STATE Append $C$ to $processed$
    \STATE Remove $C$ from $active$
\ENDFOR

\RETURN $constraints$
\end{algorithmic}
\end{algorithm}
\subsubsection{Illustration} \label{ex: constraints}
We now illustrate the space of structures bounded by constraints as guaranteed by \cref{thm: combi_data}, employing  Algorithms \ref{algo: round_update} and \ref{algo: constraints}. Continuing with the setup of Example \ref{ex: structures}, we write constraints for structure ($A>B>C>D$, [L, W, L, W]), i.e. with sequence 6:
\begin{enumerate}
    \item First round: D gets eliminated. No candidate gets more than $Q$ and $D$ gets the lowest.
    \[Q >  V_{A}^1, V_{B}^1, V_{C}^1 > V_{D}^1\]
    \item Second round: A wins after D's votes are transferred to other candidates. A gets more than the remaining candidates as well as $Q$.
    \[V_{A}^1+V_{D,A}^2 >  V_{B}^1+V_{D,B}^2, V_{C}^1+V_{D,C}^2, \ Q\]
    \item Third round: C gets eliminated. A's surplus gets distributed amongst B and C, including the surplus corresponding to D's transfer to A. Both B and C receive less than Q, and C has the lowest.
    \begin{align*}
        Q > V_{B}^1+V_{D,B}^2 + ( V_{A,B}^2 + V_{A,D,B}^3 + V_{D,A,B}^3 ) (1- \frac{Q}{V_{A}^1+V_{D,A}^2} )> \\
    V_{C}^1+V_{D,C}^2 + ( V_{A,C}^2 + V_{A,D,C}^3 + V_{D,A,C}^3 ) (1- \frac{Q}{V_{A}^1+V_{D,A}^2} )
    \end{align*}
    
\end{enumerate}

As the constraints show, the non-linearity is introduced when the surplus corresponding to a winning candidate is carried over to the next round. As far as this set of non-linear constraints is satisfied, the election outcome is fixed at ($A>B>C>D$, [L, W, L, W]).

\subsection{Characteristics of the structure framework}

Although the structures specify round-winning and losing candidates, they don't explicitly define election-winning candidates. When a significant number of ballots are partially filled, no candidates may surpass the Droop quota, leading to the last  $k$ eliminated candidates becoming the winners, as they occupy the first $k$ positions in the order. %
Lemma \ref{lem: round-election} elaborates on the relation between election-winning candidates and round-winning candidates within the mechanism, and similarly on the losing candidates. %

\begin{restatable}{lemma}{lemroundelection} \label{lem: round-election}
    A round-losing, i.e., an eliminated candidate may be an election winner, but a round-winning candidate is never an election-loser.
\end{restatable}

\begin{proof}
    We first show that round-winning candidates are always election-winning candidates. Assuming the contrary, there must be a candidate who wins a round after $k$ candidates win the rounds, i.e., after the first $k$ positions in the order are filled in $f$. Every win eliminates at least $Q$ votes from the contest (more than $Q$ if some ranked votes are not transferable, being partially ranked).  Then the number of votes in the contest is at most (total number of ballots) - (total number of candidates who won a round earlier $\times \  Q$), i.e., $T - k \left( \frac{T}{k+1}+1 \right) = T \left(\frac{1}{k+1}\right) - k$, which is always lesser than $Q = T \left(\frac{1}{k+1}\right)+1$. No candidate can then secure a win, providing contradiction.

For round-losing candidates, the meanings always coincide when all ballots are completely filled, i.e., full-length, but not necessarily otherwise. We show this through contrapositives: if a candidate is election-winning, then they are also round-winning, i.e., not round-losing.  The only way this doesn't happen is when $n-k$ candidates lose (all occupying the bottom $n-k$ positions in the partially filled order $f$), and the next round-losing candidate has to get a position in the top $k$. Assume $w$ candidates have won until now, filling $w$ top positions. Now, exactly $T-wQ$ votes remain in the contest when all surplus votes are transferable, i.e.,  all ballots are full-ranked, allowing transfers. If not, less than $T-wQ$ votes remain. If the current round has an elimination, then all remaining $k-w$ candidates must have less than $Q$ votes. This means the total votes in the system are at most $(k-w)Q$. But, $(k-w)Q>T-wQ$ implies  $Q = \left(\frac{T}{k+1}+1\right) > T/k$, giving a contradiction. Thus, we show that a round-losing candidate may be an election winner, but a round-winning candidate can never be an election-losing candidate. 
\end{proof}

\section{Algorithmic Framework: Details}

\subsection{Optimizing to reach an outcome structure}\label{app:allocation_algo}
\begin{algorithm}
\caption{Allocation Rule} \label{algo: allocation}
\begin{algorithmic}[1]
\REQUIRE  $\mathcal{C}$, voter data $\{V\}_{\mathcal{C}}$, Constants $B$, $Q$
\ENSURE Budget allocation, under feasibility
\STATE \textbf{function} \textsc{SmartAllocation}( $\mathcal{C}$, $\{V\}_{\mathcal{C}}$, $B$, $Q$)
  \STATE Initialize $\text{votes}, \text{ votes}_{\text{new}} \leftarrow 0$.

  \FOR{processing each round}
    \STATE $C_s \leftarrow $ desired candidate for elimination or win in this round
      \IF{$C_s$ needs elimination}
       \IF{all in-contest candidates have votes $< Q$}
        \STATE Add $\text{votes}_{\text{new}} \ \forall  \ C_i : \  V^1_{C_i} + \text{votes}_{\text{new}} \geq V^1_{C_s} + 1$.
        \ELSE
        \RETURN infeasible
        \COMMENT{Increase $Q$ (i.e.,  budget $B$).}
        \ENDIF
        \ENDIF
      \IF{$C_s$ needs to win}
       \IF{all in-contest candidates have votes $< Q$}
        \STATE Add $\text{votes}_{\text{new}}\text{ to } C_s : \  V^1_{C_s} + \text{votes}_{\text{new}} \geq Q +1$.
        \ELSE
        \STATE Add $\text{votes}_{\text{new}} \text{ to } C_s : \  V^1_{C_s} + \text{votes}_{\text{new}}\geq \max(V^1_{C_i}) + 1$.
        \ENDIF
        \IF{ $C_{last}$ doesn't transfer to $C_s$}
        \STATE Add $\text{votes}_{\text{new}}  [C_{last}, C_s] : \  V^1_{C_{last}}\leftarrow V^1_{C_{last}} + 1$.        \IF{$margin(V^1_{C_{last}})< 1$} 
        \RETURN infeasible \COMMENT{Increase $Q$ (i.e.,  budget $B$).}
        \ENDIF
        \ENDIF
      \ENDIF
      \STATE  $\text{votes} \leftarrow \text{votes} + \text{votes}_{\text{new}}$ 
      \STATE $\mathcal{C}\gets \mathcal{C}/ C_s$, update  $\{V\}_\mathcal{C} $
      
\ENDFOR
      \IF{votes $\geq$ budget}
       \RETURN feasible \COMMENT{Distribution of added votes}
       \ELSE
       \RETURN infeasible \COMMENT{Increase budget $B$}
    \ENDIF
\end{algorithmic}
\end{algorithm}

\polyefforts*
\begin{proof}

We now describe the allocation rule (Algorithm \ref{algo: allocation}) that is implemented for filling slacks, starting with round 1. While allocating, we break the ties in favor of the candidate who receives votes. Even if candidates have fractional votes from surplus transfers, the additions are integers. For each round, the structure mandates either an elimination or win of a candidate $C_s$, forming two cases:
\begin{itemize}
    \item Case-(L) $C_s$ needs to be eliminated: If $max(C_i)<Q $, i.e., all candidates have less than the quota, add appropriately to $C_i$ such that $C_i \geq C_s+1 $ is satisfied. If $max(C_i)\geq Q$ for any $C_i\in \mathcal{C}$ or $C_s+1\geq Q$, we reach infeasibility and the structure may become feasible only by increasing $Q$, i.e. $B$.
    \item Case-(W)  $C_s$ needs to win: If $max(C_i)<Q $, i.e., all candidates have less than the quota, add to $C_s$ such that $C_s \geq Q +1$ is satisfied. Otherwise if $max(C_i)\geq Q$, we want $C_s \geq max(C_i)+1 $ to be satisfied. 
    In both, we check if the elimination/win before the current round, say of $C_{last}$, allows at least a vote transfer to $C_s$. 
    If so, nothing to do more. If not, we allocate at least one ranked vote [$C_{last}, C_s$] to $C_{last}$, so that $C_s$ doesn't cross $Q$ until the current round. This addition doesn't affect $C_{last}$'s position if  $margin(C_{last})$  ( i.e., the difference between the bottom two or $V_{C_{last}}$ - $Q$ at the time of $C_{last}$'s elimination or win respectively) is positive, given non-degeneracy. However, if $Q-V_{C_{last}} < 0$, we hit infeasibility. 
    This makes sure that only after $C_{last}$ leaves the contest, $C_s$ satisfies $C_s\geq max(C_i)+1$. In other words, this avoids the case where a candidate $C$ receives a budget to win a round and when the election proceeds on this updated voter data after additions, they win a prior round. 

    \end{itemize}
In both these cases, if there are any `available' votes, i.e., votes that were added in previous rounds and their beneficiaries are out of the contest now, we reuse those by adding subsequent preferences to them. For example, if we add $X$ ballots to $C_1$, i.e., [$C_1$], upon elimination of $C_1$, we may reuse (all $X$ or partial) while adding to $C_4$, i.e., updating the added ballots to $[C_1, C_4]$.

We now show the efficacy of the allocation rule by satisfying the two conditions listed in the proof sketch. First, see that if the process converges in only one iteration without creating new slacks (i.e., when (b) is true), there is no way of reaching the structure by spending less, thus satisfying (a). This follows as the slacks are greedily filled with the minimum required votes to satisfy the constraints in each round. If a candidate receives any allocation, it is available for reuse right after the candidate is out so that every added vote is either in use or available until the final round, which already has a deterministic result with only one candidate remaining. 

More formally, the optimality is justified using the induction argument on the number of rounds receiving allocations: for the base case, we assume only one round receives additions. Naturally, the greedy allocation on this round ensures optimality, as the structure fixes all prior round transfers. Next, assume that when our allocation happens over $r$ rounds, the optimality is guaranteed. Considering the $r+1$ case now, the first round (needing additions) can be optimally allocated and the election instance (including the first-round allocation votes) may be reduced by processing its win/elimination result. See that this allocation is independent of the later rounds. After the first round, no stock of `available' votes is generated (the votes given for not losing are still with the candidates, and votes given for winning are just enough to win, not producing any surplus), and the allocation is again determined on the reduced election instance. As the allocation rule itself is independent of the stock of `available' votes, apart from viewing it as a first-preference budget pile, the resulting optimal allocation for $r$ rounds is exactly the same as for $r+1$ rounds on the original instance. Thus, (a) holds.

For (b), we argue that under feasibility, each round needs to be processed exactly once using the specified allocation rule. Since slacks are filled consecutively, the effect of round $j$ additions on round $j+1$ onward is already taken care of. We further show that this cannot create slacks for $<j$ rounds as well. See that in Case-(L), only in-contest candidates are allotted votes in L rounds- If $C_i$ receives any additions in $j$th L round, it implies it had $<Q$ until $j$th round, so it doesn't win prior to it. It also doesn't lose prior to it because no candidates that lose before round $j$ receive any additions post their eliminations, keeping the elimination order intact.  For Case-(W), additions are made to only one candidate $C_s$ at any time. Since $C_s$ cannot win earlier by the allocation rule, no new slacks are generated under feasibility and the algorithm converges in a single iteration, taking care of (b) as well.

\noindent \textbf{Complexity: } Algorithm  \textsc{SmartAllocation} iterates over all rounds, i.e. over the number of candidates $n$, and in each, it checks all in-contest candidates and updates the active votes after transfers. The addition operations are straightforward, making overall $O(n)$ complexity per round. The active votes update has $O(m)$ complexity for $m$ number of unique ballots, as at most all ballots see transfers. Then, for $n$ rounds, the complexity becomes $O(mn+n^2)$, i.e. $O(mn)$.

\end{proof}

\subsection{Efficiently navigating the search space}\label{app: cut_search_space}

\begin{algorithm}
\caption{Reducing Election Instance}
\label{algo:elimination_subfunctions}
\begin{algorithmic}[1]
\STATE \textbf{function} \textsc{ReduceElectionInstance}($L, \{D\}_I$)
\FORALL{orders $v$ and values $value$ in $\{D\}_I$}
    \IF{$value > 0$}
        \STATE $v' = [C_j \in v \text{ if } C_j \notin L$] maintaining the same order as $v$
        \STATE Update $\{D\}_I$: $\{D\}_I[v'] \mathrel{+}= value$
    \ENDIF
\ENDFOR
\end{algorithmic}
\end{algorithm}

\begin{algorithm}
\caption{Computing Strict-Support}
\label{algo:strict_support}
\begin{algorithmic}[1]

\STATE \textbf{function} \textsc{Strict-Support}($L, G, \{D\}_I$)
\STATE Initialize dictionary $S$ with keys from set $L$ and values as zeros.
\FORALL{orders $v$ and values $value$ in $\{D\}_I$}
    \IF{$v[0] \in L$ and $value > 0$}
        \STATE $v' = [C_j \in v \text{ until } C_j \in G$]
        \STATE Update $S$: $S[C_i] \mathrel{+}= value$ if $C_i \in v'$
    \ENDIF
\ENDFOR
\RETURN S
\end{algorithmic}
\end{algorithm}

\begin{algorithm}
\caption{Removal of irrelevant candidates under uncertainty}
\label{algo:candidate-removal}
\begin{algorithmic}[1]
\REQUIRE  $\mathcal{C}$ in increasing order of $V^1$, Dictionary of orders (ballots) and values (votes for those ballots) $\{D\}_I$, Constants $B$, $Q$
\ENSURE Smallest relevant candidate set $\mathcal{C}$

\STATE \textbf{procedure} \textsc{IrrelevantCandidateRemoval}($\mathcal{C}, \{D\}_I, B, Q$)
\STATE $L, S \gets  \{\}, $  \text{dictionary with keys from set $\mathcal{C}$ and values as zeros}
\WHILE{$B+S[C_i]<Q$ for all {$C_i \in L$} \AND $\mathcal{C} \neq \phi$}
    \STATE Shift first  candidate $C_1\in \mathcal{C}$ to set $L$
    \STATE $S \gets \textsc{Strict-Support}(L, \mathcal{C}/L, \{D\}_I$)
    \STATE $S_j^i  \gets \textsc{Strict-Support}(C_j\cup L/C_i, \{\}, \  \{D\}_I) \ \forall C_i, C_j \in L, \mathcal{C}/L $
    \IF{$ B+ S[C_i] < S_j^i[C_j]< Q  \ \ \forall \ C_i, C_j \in L, \mathcal{C}/L$} \label{line: ifconditions}
        \STATE $\{D\}_I \gets \textsc{ReduceElectionInstance}(L, \{D\}_I$)
        \STATE $\mathcal{C}, L \gets \mathcal{C}/L, \{\} $ 
        \STATE Sort $\mathcal{C}$ in increasing order of $V^1$
    \ENDIF
\ENDWHILE
\RETURN $\mathcal{C}, \{D\}_I$
\end{algorithmic}
\end{algorithm}

\candidateremoval*
\begin{proof}
As detailed in the proof sketch, we first prove Lemma \ref{lem: set_reduction}.

\begin{lemma} \label{lem: set_reduction}
    Given all eliminations until the $k$'th round, the structure information until round $k$, i.e., the elimination order detail is irrelevant for the later rounds. 
\end{lemma}
\begin{proof}
    Let $L$ be the set of $k$ candidates that are all eliminated in a specific order. We now shrink the vote base by transferring all first-choice votes of $L$ to their subsequent choices in $\mathcal{C}/L$. Since there is no limit on the number of transfers a vote can make, and there isn't a win so far leading to a nonzero weight on transferable votes, the reduced voter base is independent of the order in which candidates in $L$ got eliminated. All the eliminated candidates occupy the last $k$ positions in the order, i.e., the social choice order $f$. For any complete structure, the constraints for the later rounds involve the exact same variables, aggregated from previous rounds.  Note that the same may not hold if the history of $k$ rounds involves any win, as then the position of the win governs the weight of surplus votes getting transferred further, immediately affecting all future rounds. 
\end{proof}

Following Lemma \ref{lem: set_reduction}, if we remove the set of candidates occupying strictly consecutive lower positions in the social choice order, and run the election for the smaller set, the remaining candidates will appear in the same order as before. This doesn't hold if we remove candidates at random as their removal may have ripple effects originating from the prior round results.  If we can remove irrelevant candidates, the updated election instance may be obtained in Algorithm \ref{algo:elimination_subfunctions}.

With an allowed budget of $B$ ranked-choice votes to be added, we now identify certain conditions on the current voter data so that election instance is reduced by dropping the set of projected eliminations. For that, we lay out a few definitions for counting votes when an election instance is constrained to a group of candidates.

\begin{definition} \label{def: strict_support}
    $\text{Strict-Support}_{L, G}(C_i)$ is the number of votes candidate $C_i$ receives from voters that rank candidates in $L$ the first, excluding votes that rank $C_i$ post any of $G$ candidates. Here, $L$ and $G$ are mutually exclusive sets. Accordingly, $\text{Strict-Support}_{L, \{\}}(C_i)$ becomes the total number of votes candidate $C_i$ receives (at any rank) from voters that rank candidates in $L$ the first.
\end{definition}

Using the above definition, computing Strict-Support may be achieved straightforwardly using \Cref{algo:strict_support}.

Based on the ideas explained in the proof sketch, the candidates' removal may now be obtained using Algorithm \ref{algo:candidate-removal}. We construct a set of candidates that all get eliminated before any win is secured, despite an uncertain addition of $B$ votes. %
As each candidate gets identified for elimination, we gradually update the criterion for such identification. Finally, using Lemma \ref{lem: set_reduction}, the set of candidates and the election instance can be reduced. Note that the droop quota is kept unchanged, despite the reduction. We assume without loss of generality-since any candidates with greater than $Q$ votes occupy the first positions of the social choice order- the election instance is already reduced after transferring any of their surpluses to candidates still in the contest.

\textbf{Functioning of the algorithm:} The algorithm has two components: a \textbf{while} condition for group building + book-keeping of important variables and an \textbf{if} condition for executing elimination. The algorithm terminates (via the while condition) when strict-support of the group under consideration surpasses the quota $Q$. Note that the termination is guaranteed with $L= \mathcal{C}$ as the trivial last iteration. 

    The \textbf{while} condition: We ensure that all candidates in $\mathcal{C}$ have votes less than $Q$, so that no one crosses the quota for winning. We need to ensure that if candidates in $L$ follow any order of elimination, the last candidate (with all gained transfer votes via eliminations) has still fewer votes than the remaining candidates. For this, we employ the concept of strict-support for a group of candidates $L$, against $\mathcal{C}/L$ (the first call to strict-support). %
    It gives the number of votes a candidate can receive if it is the last one to get eliminated from the group. Notably, it doesn't count the number of votes that include transfers via candidates in $\mathcal{C}/L$.
    While the strict-support of all candidates in $L$ (plus $B$) is smaller than $Q$, i.e., none crosses the Droop quota, we add candidates to $L$, even if the \textbf{if} condition is not satisfied.
    
    The \textbf{if} condition: This ensures that the election instance is reduced only when all candidates in the $L$ have [strict-support + $B$] less than [the first choice votes + any votes through transfers] for all candidates $i\in \mathcal{C}/L$. To capture the transfers to candidates outside $L$, we again use strict-support. $S_j^i[C_j]$ equals the first choice votes $C_j$ has + the amount of votes transferred to $C_j$ from all of $L$ except $i$, i.e., when $i$ is the last candidate to be eliminated in $L$. The group building continues until we are able to reach this reduction condition, or if the algorithm is terminated. After removal, we update the first-choice votes of existing candidates, i.e., $V_i^1$, where first-choice votes of the eliminated candidates are transferred to their subsequent available choices, if any. The algorithm restarts again after any such removal.

  \noindent  \textbf{Complexity: } First, see that the complexity of the function \textsc{Strict-Support} is $O(mn)$ as it iterates over all unique ballot types $m$ to find the Strict-Support of candidates in $L$. The complexity of the function \textsc{ReduceElectionInstance} is similarly $O(m)$ as this updates the values of at most all ballot types. Now, consider the complexity of \textsc{IrrelevantCandidateRemoval}: the \textbf{while} loop is run at most $n$ times, computing \textsc{Strict-Support} $n+n^2$ times. Additionally, whenever the \textbf{if} loop is run, it calls \textsc{ReduceElectionInstance} and sorts $\mathcal{C}$ with $n\log n $ complexity. Overall, the complexity of \textsc{IrrelevantCandidateRemoval} is then $O(n((n+n^2)\times mn + m + n \log n))$, i.e. $O(mn^4)$.
\end{proof}
\substtheorem*
\begin{proof}
We show that the sequence information is extracted from a given election instance by finding bounds on both the number of total wins and the number of losses in the row.
We use Algorithm \ref{algorithm:upper_bound_W} to find the bounds as follows:
\begin{enumerate}
    \item Upper bound on the number of $W's$ ($U_W$): We first find the set of candidates $C_W$ who have $\text{Strict-Support}_{\mathcal{C}, \{\}}(C_i)+B$ greater than $Q$. It equals the number of candidates $n$ only when all ballots are fully ranked. The upper bound on the number of wins, $U_W$ is then equal to [$B$ + the total number of unique ballots in $C_W$ $/Q$]. The number of unique ballots may be straightforwardly computed in a loop over $\text{Strict-Support}_{\mathcal{C}, \{\}}(C_W)$.  Note that this number is also naturally bounded by $k$.  
    \item Lower bound of the number of consecutive initial losses ($i_L$): After counting $U_W$, we estimate the range where $W's$ may be placed in sequences, again under uncertainty of $B$ additional votes. We first define the definitely losing candidates $C_L$ as those with $B+\text{Strict-Support}_{\mathcal{C}, \{\}}(C_i)$ less than at least $k$ candidates' first choice votes, i.e., $V_i^1$ for top $k$. We then sort $C_L$ in decreasing order of their transferable votes, i.e., $t_i = \text{Strict-Support}_{C_i,\{\}}(\mathcal{C}/C_L)$: These are the votes candidate $C_i$ eventually transfers to a candidate outside $C_L$. Finally, we move the transferable votes $t_i$ from $C_L$ to the candidate $C_{top}$ with the highest first-choice votes $V_{top}^1$, while $\sum t_i + V_{top}^i< Q$ remains satisfied. The number of such transfers gives a lower bound on the number of consecutive initial losses.
\end{enumerate}
\begin{algorithm}
\caption{Reducing the space of sequences}
\label{algorithm:upper_bound_W}
\begin{algorithmic}[1]
\REQUIRE Set of candidates $\mathcal{C}$, $B, \ Q$, $\textsc{Strict-Support}$
\ENSURE  Upper Bound on the number of $W$'s, i.e., $U_W$; Lower bound on the number of consecutive initial losses $i_L$
    \STATE \textbf{function} \textsc{Predict-Wins}($\mathcal{C}, B, Q$)
    \STATE $C_W \gets \{C_I \subset \mathcal{C} \textbf{ if } \text{Strict-Support}_{\mathcal{C}, \{\}}(C_i) + B > Q\}$
    \STATE $\bar{C_W}, \text{number-unique-ballots}_{C_W} \gets \{\}, 0$   
    
    \FOR{each candidate $C_j$ in $C_W$}
        \STATE $\text{number-unique-ballots}_{C_W} \gets \text{number-unique-ballots}_{C_W} + \text{Strict-Support}_{\mathcal{C}, \bar{C_W}}(C_j)$
        \STATE $\bar{C_W} \gets \bar{C_W} + C_j$ 
    \ENDFOR
 
    \STATE \textbf{return} $U_W  = \min (k, [(B+ \text{number-unique-ballots}_{C_W}) / Q])$

\vspace{0.35cm}
     \STATE \textbf{function} \textsc{Predict-Losses}($\mathcal{C}, B, Q$)
    \STATE $\text{[first-choice] } \gets  \{V_i^1\} \text{ sorted in decreasing order} $ 
    \STATE $C_L \gets \{C_I \subset \mathcal{C} \textbf{ if }  \text{Strict-Support}_{\mathcal{C}, \{\}}(C_i) + B < \text{[first-choice ]}[k] \} $
    \STATE $T \gets \{t_i = \text{Strict-Support}_{C_i,\{\}}(\mathcal{C}/C_L), \ C_i \in C_L \} $ sorted in decreasing order
    \STATE Initialize $i_L=1$
    \WHILE{$\sum T[:i_L] + \text{[first-choice ]}[1]<Q$}
    \STATE $i_L=i_L+1$
    \ENDWHILE
\RETURN  $i_L$      
\end{algorithmic}
\end{algorithm}
These together result in a smaller space that covers all feasible sequences given an uncertain addition of $B$ votes.

\noindent
\textbf{Complexity:} Function \textsc{Predict-Wins} first calls  \textsc{Strict-Support}
$n$ times, and then runs a straightforward loop over at most $n$ candidates. This amounts to $O(mn^2+n)$, i.e. $O(mn^2)$ complexity. The function \textsc{Predict-Losses} calls \textsc{Strict-Support} at most $2n$ times and later runs a simple loop for computing the bound on losses. This amounts to the complexity of $O(mn^2)$.

\end{proof}
\section{2024 US Republican Primaries Case Study: Details}\label{app: casestudies}
The FairVote and WPA Intelligence ranked-choice poll sampled likely Republican presidential primary voters nationwide after the second Republican presidential debate. The results are reproduced in Table \ref{tab:gop_primary}.  This survey encompassed a broad national scope, with responses collected from across the country. To ensure national representativeness, the data was weighted appropriately for each state, resulting in some fractional votes. Our analysis adheres to the same methodology and settings as those detailed in \cite{OtisLaverty2023}.

\begin{table}[ht]
\centering
\resizebox{\textwidth}{!}{%
\begin{tabular}{|l|l|l|l|l|l|l|l|l|l|l|l|l|l|}
\hline
 & Candidates & 1 & 2 & 3 & 4 & 5 & 6 & 7 & 8 & 9 & 10 & 11 & 12 \\ \hline
T & Trump & 47.6\% & 47.6\% & 47.6\% & 47.6\% & 47.7\% & 47.7\% & 47.7\% & 47.9\% & 48.4\% & 48.7\% & 54.5\% & 62.3\% \\ \hline
H & Haley & 7.7\% & 7.9\% & 7.9\% & 8.1\% & 8.4\% & 8.5\% & 8.9\% & 10.3\% & 11.9\% & 20.3\% & 24.8\% & 37.7\% \\ \hline
R & Ramaswamy & 12.6\% & 12.6\% & 12.6\% & 12.8\% & 12.9\% & 12.9\% & 12.9\% & 13.9\% & 14.8\% & 15.8\% & 20.7\% &  \\ \hline
D & DeSantis & 12.7\% & 12.7\% & 12.7\% & 12.7\% & 12.8\% & 12.9\% & 13.1\% & 13.2\% & 14.1\% & 15.2\% &  &  \\ \hline
C & Christie & 8.5\% & 8.5\% & 8.5\% & 8.5\% & 8.5\% & 9\% & 9.2\% & 9.9\% & 10.8\% &  &  &  \\ \hline
P & Pence & 4.4\% & 4.4\% & 4.4\% & 4.4\% & 4.5\% & 4.6\% & 4.7\% & 4.9\% &  &  &  &  \\ \hline
Sc & Scott & 2.8\% & 2.8\% & 2.9\% & 2.9\% & 2.9\% & 2.9\% & 3.5\% &  &  &  &  &  \\ \hline
Hr & Hurd & 1\% & 1\% & 1\% & 1.2\% & 1.2\% & 1.5\% &  &  &  &  &  &  \\ \hline
Ht & Hutchinson & 1.1\% & 1.1\% & 1.1\% & 1.1\% & 1.1\% &  &  &  &  &  &  &  \\ \hline
Y & Youngkin & 0.8\% & 0.8\% & 0.8\% & 0.8\% &  &  &  &  &  &  &  &  \\ \hline
B & Burgum & 0.4\% & 0.4\% & 0.5\% &  &  &  &  &  &  &  &  &  \\ \hline
E & Elder & 0.3\% & 0.3\% &  &  &  &  &  &  &  &  &  &  \\ \hline
Su & Suarez & 0.1\% &  &  &  &  &  &  &  &  &  &  &  \\ \hline
\end{tabular}
}
\caption{RCV on primary- WPA intelligence and FairVote poll %
}
\label{tab:gop_primary}
\end{table}

\subsection{Strategies under perfect information} \label{app: perfect_info}
We first describe the STV structure framework that we adopt to optimize over the Republican primary data, which involves 13 candidates. This results in (13!) possible social choice orders. The poll data corresponds to the order of [T,H,R,D,C,P,Sc,Hr,Ht,Y,B,E,Su]. Our focus is on devising strategies for the top two candidates, but since the surplus votes of the winner are not transferred to the remaining candidates, we maintain the Droop quota at 800.  This practically allows only one sequence, [$L,L,L,..L,W$], where the winner and loser of the final round become the first and second top candidates, respectively. 

To achieve any specific social choice order, we may now set up corresponding constraints using \Cref{thm: combi_data}. We then solve the optimization problem efficiently via Theorem \ref{thm: poly_efforts_B}. Theoretically, finding optimal additions for both winning the election and placing in the top two would require optimizing for all 13! total orders in a general context. However, we apply methods from Section \ref{sec: cut_search_space} to answer more efficiently, thereby significantly reducing computational complexity.

\subsubsection{Winning the election (Table \ref{tab: votestowin}):}
To find the strategic additions needed for winning this election, we first analyze how close the candidates can get to candidate A.  For this, we use the concept of $\text{Strict-Support}_{L, G}(C_i)$, which according to Definition \ref{def: strict_support}, is the number of votes candidate $C_i$ receives from voters that rank candidates in $L$ the first, excluding votes that rank $C_i$ post any of $G$ candidates. Let $\mathcal{C} = \{T,H,R,D,C,P,Sc,Hr,Ht,Y,B,E,Su\}$ be the set of all candidates in the election. We now calculate the \text{Strict-Support} with $L = \{\mathcal{C}/T\}$ and $G = \{T\}$, for $C_i$ equal to all candidates except $T$, as given in Table \ref{tab: strict_support_primary}. For each candidate $C_i$, this gives the percentage of votes $C_i$ has when it reaches the last round alongside T, i.e., this is the maximum support $C_i$ can have opposite T.

\begin{table}[ht]
    \centering
\begin{tabular}{|c|c|c|c|c|c|c|c|c|c|c|c|c|}
\hline
H & R & D & C & P & Sc & Hr & Ht & Y & B & E & Su \\
\hline
37.8 & 41.3 & 43.8 & 32.0 & 34.7 & 36.1 & 29.6 & 30.4 & 32.0 & 31.1 & 30.2 & 29.3 \\
\hline
\end{tabular}
    \caption{Strict-Support of candidates H to Su, against Trump (in percentages)}
    \label{tab: strict_support_primary}
\end{table}

Comparing with the primary election data in Table \ref{tab:gop_primary}, see that the first choice support of T ($47.6\%$) is higher than the Strict-Support of all candidates in $\mathcal{C}/T$, against T. Then, for any elimination order of H to Su and thus for any candidate $C_i$, see that the following holds: For candidate $C_i$ to win, a minimum of $\text{Strict-Support}_{\mathcal{C}/\{C_i\}, C_i}(T)$ (i.e., T's votes when only $C_i$ is not eliminated) - $\text{Strict-Support}_{\mathcal{C}/\{T\}, T}(C_i)$ (i.e., $C_i$'s votes when only T is not eliminated) votes must be added to $C_i$. Further, if this minimum (selfish) addition can also guarantee that $C_i$ doesn't get eliminated until the last round, then this becomes the optimal addition for $C_i$ to win. 

We then calculate the minimum required additions for candidates T to C, which Table \ref{tab: votestowin} presents as `Strategic additions'. These are the differences in Strict-Support for T and $C_i$, presented as the percentages in column `Head to Head against T'. For each of these additions made to $C_i$ in $\{H, R, D, C\}$, we may refer to Table \ref{tab:gop_primary} to confirm that no $C_i$ gets eliminated until the last round against T.  This concludes the optimality of the results in Table \ref{tab: votestowin}.

\subsubsection{Being in the top two (Table \ref{tab: top2strategies_5}):}

To solve this optimally, we use a two-stage approach. First, we remove the bottom 8 candidates from the election instance and transfer their votes to the remaining candidates, who will always occupy the first 5 positions in the social choice order (Algorithm \ref{algo:candidate-removal}). This establishes that none of these candidates can be a part of the top 5 even if up to $5\%$ unknown ranked ballots are added. Second, we iterate over all possible orders where (T, $C_i$) are the top two candidates for $C_i \in \{H, R, D, C\}$, and find optimal additions to reach each of those using Theorem \ref{thm: poly_efforts_B}. For each pair (T,$C_i$), we choose the order that corresponds to the minimum addition, as presented in Table \ref{tab: top2strategies_5}.   

\begin{table}[ht]
    \centering
    \begin{tabular}{|c|c|c|c|c|c|c|c|c|}
\hline
& P & Sc & Hr & Ht & Y & B & E & Su \\
\hline
$S[C_i]$ & 4.92 & 4.35 & 2.26 & 2.27 & 1.89 & 1.56 & 1.08 & 1.07 \\
\hline
$\min_{C_j} S_j^i[C_j]$ & 9.92 & 9.91 & 10.19 & 9.72 & 10.32 & 10.51 & 10.42 & 10.70 \\
\hline
Difference & 5.00 & 5.56 & 7.93 & 7.45 & 8.43 & 8.95 & 9.34 & 9.63 \\
\hline
\end{tabular}
    \caption{Strict-Support of P to Su, against T to C (in percentages)}
    \label{tab: strict-supportFtoM}
\end{table}

For candidates' removal, we show that the \textbf{if} condition in Algorithm \ref{algo:candidate-removal} (line \ref{line: ifconditions}) is satisfied for the budget of $5\%$ for $L=\{P,Sc,Hr,Ht,Y,B,E,Su\}$. The \textbf{while} condition is trivially satisfied by large $Q$. Note that $L$ is the set of bottom 8 candidates in Table \ref{tab:gop_primary}, as built by this condition. The Strict-Support of candidates in $L$ ($S[C_i]$ in line \ref{line: ifconditions}) and minimum support of $G = \{T, H, R, D, C\}$ against $L$ ($\min_{C_j} S_j^i[C_j]$ in line \ref{line: ifconditions}) is given in Table \ref{tab: strict-supportFtoM}. The minimum difference of $5\%$ shows that Algorithm 2 can successfully remove the bottom 8 candidates for up to $5\%$ additions.

\subsection{Robustness under uncertainty}
We created 1000 bootstrap datasets using the original poll data, each containing the same number of voters. To analyze each dataset, we use the methodology of applying \Cref{thm: remove_irrelevant_candidates} as in \Cref{app: perfect_info}. In addition, if the removal condition isn't satisfied for the bottom group $L$, we also check if the budget is enough to save the candidate with the highest Strict-Support from $L$, from elimination up to the next round, thereby checking if the budget can affect two eliminations simultaneously.  For up to 5$\%$ additions,  this together removes 7 or 8 irrelevant candidates for $98\%$ of the samples, successfully allowing us to find optimal strategies for all applicable candidates. Note that the set of irrelevant candidates may be different for each sample, as candidates P and Sc also find strategies in some instances (Table \ref{tab: bootstrap_summary_table}). We also analyze the strategies for up to $3\%$ and $4\%$ additions, which reaffirms that H is the most stable second top and that R and D get weaker with decreasing allowed budget.
\subsubsection{Analysis for up to $4\%$ strategic additions}
For up to $4\%$ additions, Theorem \ref{thm: remove_irrelevant_candidates} removed irrelevant candidates in $99.7\%$ of the samples, giving us the strategy analysis in Table \ref{tab: strategy_4_1} and \ref{tab: strategy_4_2}.
\begin{table}[ht]
\begin{tabular}{|l|l|l|l|l|l|l|}
\hline
Candidate & H & R & D & C & P & Sc \\ \hline
Top 2 frequency in bootstrap samples (\%) & 68.51 & 23.27 & 5.12 & 1.40 & - & - \\ \hline
Top 2 frequency under 5\% strategic additions (\%) & 99.10 & 84.25 & 60.28 & 31.90 & 0.50 & 0.10 \\ \hline\hline
Average additions in strategy (\%) & 1.20 & 1.80 & 2.31 & 2.53 & 3.72 & 4.00 \\ \hline
\end{tabular}
\caption{For strategic additions up to $4\%$}
\label{tab: strategy_4_1}
\end{table}

\begin{table}[ht]
\centering
\setlength{\tabcolsep}{5pt} %
\renewcommand{\arraystretch}{1.2} %
\begin{tabular}{|c|c||c|c||c|c||c|c||c|c||c|c|}
\hline
\multicolumn{2}{|c||}{Haley} & \multicolumn{2}{c||}{Ramaswamy} & \multicolumn{2}{c||}{DeSantis} & \multicolumn{2}{c||}{Christie} & \multicolumn{2}{c||}{Pence} & \multicolumn{2}{c|}{Scott} \\
\hline
Type & \SI{}{\percent} & Type & \SI{}{\percent} & Type & \SI{}{\percent} & Type & \SI{}{\percent} & Type & \SI{}{\percent} & Type & \SI{}{\percent} \\
\hline
H   & \SI{90.82}{} & C   & \SI{43.42}{} & DC   & \SI{48.00}{} & C   & \SI{62.50}{} & P   & \SI{60.0}{} & Sc  & \SI{100.0}{} \\
HR  & \SI{2.95}{}  & R   & \SI{26.81}{} &  D  & \SI{25.64}{} & DC  & \SI{27.96}{} & CP  & \SI{20.0}{} &     &              \\
HD  & \SI{2.30}{}  & RC  & \SI{19.41}{} & C   & \SI{20.00}{} & CR  & \SI{4.93}{}  & DP  & \SI{20.0}{} &     &              \\
HC  & \SI{2.30}{}  & DC  & \SI{5.76}{}  & RC  & \SI{1.82}{}  & D   & \SI{2.63}{}  &     &              &     &              \\
D   & \SI{1.64}{}  & RDC & \SI{2.80}{}  & RDC & \SI{1.64}{} & RDC & \SI{0.99}{}  &     &              &     &              \\
    &              & RD  & \SI{0.82}{}  & RD  & \SI{0.55}{} & HC  & \SI{0.66}{}  &     &              &     &              \\
    &              & D   & \SI{0.66}{}  & HD  & \SI{1.27}{} & DCSc & \SI{0.33}{}  &     &              &     &              \\
    &              & HR  & \SI{0.33}{}  & DCP & \SI{0.36}{} &     &              &     &              &     &              \\
    &              &     &             & DSc & \SI{0.18}{} &     &              &     &              &     &              \\
\hline
\end{tabular}
\caption{Strategy categorization under uncertainty, for up to $4\%$ strategic additions}
\label{tab: strategy_4_2}
\end{table}

\subsubsection{Analysis for up to $3\%$ strategic additions}
For up to $3\%$ additions, Theorem \ref{thm: remove_irrelevant_candidates} removed irrelevant candidates in $99.9\%$ of the samples, giving us the strategy analysis in Table \ref{tab: strategy_3_1} and \ref{tab: strategy_3_2}.
\begin{table}[ht]
\begin{tabular}{|l|l|l|l|l|}
\hline
Candidate & H & R & D & C \\ \hline
Top 2 frequency in bootstrap samples (\%) & 68.37 & 23.22 & 5.11 & 1.40 \\ \hline
Top 2 frequency under 5\% strategic additions (\%) & 97.998 & 73.173 & 41.341 & 20.12 \\ \hline\hline
Average additions in strategy (\%) & 1.11 & 1.44 & 1.71 & 1.96 \\ \hline
\end{tabular}
\caption{For strategic additions up to 3\%}
\label{tab: strategy_3_1}
\end{table}

\begin{table}[ht]
\centering
\setlength{\tabcolsep}{5pt} %
\renewcommand{\arraystretch}{1.2} %
\begin{tabular}{|c|c||c|c||c|c||c|c|}
\hline
\multicolumn{2}{|c||}{Haley} & \multicolumn{2}{c||}{Ramaswamy} & \multicolumn{2}{c||}{DeSantis} & \multicolumn{2}{c||}{Christie} \\
\hline
Type & \SI{}{\percent} & Type & \SI{}{\percent} & Type & \SI{}{\percent} & Type & \SI{}{\percent} \\
\hline
H   & \SI{90.88}{} & C   & \SI{48.50}{} & DC   & \SI{38.12}{} & C   & \SI{60.96}{} \\
HR  & \SI{3.04}{}  & R   & \SI{27.25}{} & D   & \SI{28.45}{} & DC  & \SI{27.27}{} \\
HD  & \SI{2.36}{}  & RC  & \SI{16.03}{} & C   & \SI{27.90}{} & CR  & \SI{5.88}{}  \\
HC  & \SI{2.03}{}  & DC  & \SI{5.21}{}  & RC  & \SI{1.93}{}  & D   & \SI{4.28}{}  \\
D   & \SI{1.69}{}  & RDC & \SI{1.40}{}  & RDC & \SI{1.38}{}  & RDC & \SI{1.07}{}  \\
    &              & D   & \SI{0.80}{}  & HD  & \SI{1.10}{}  & HC  & \SI{0.53}{}  \\
    &              & RD  & \SI{0.60}{}  & R   & \SI{0.83}{}  &     &              \\
    &              & HR  & \SI{0.20}{}  & RD  & \SI{0.28}{}  &     &              \\
\hline
\end{tabular}
\caption{Strategy categorization under uncertainty, for up to $3\%$ strategic additions}
\label{tab: strategy_3_2}
\end{table}

\subsubsection{Efficacy of optimal strategies under imperfect information (under selfish modifications)}
We show additional analysis in Table \ref{tab:strategy_modifications} where the strategy for each candidate $X$ is modified by changing it from [altruistic to C] to [altruistic to C, selfish], i.e., [C, X] by ranking two preferences in the added ballots. This strictly increases the strategist's winning probability, without decreasing C's. %

\begin{table}[ht]
\centering
\begin{tabular}{|l|c|c|c|c|c|}
\hline
Strategy ($\%$) & T, H & T, R & T, D & T, C & D, T \\ \hline
Original (No Strategy) & 68.8 & 24.4 & 5.3 & 1.5 & - \\ \hline
Ramaswamy's Strategy (+[C,R]=1.25$\%$) & 43.4 & \textbf{43.1} & 8.5 & 5.0 & - \\ \hline
DeSantis's Strategy (+[C,D]=2$\%$, +D=0.87$\%$) & 31.7 & 31.6 & \textbf{27.5} & 9.1 & 0.1 \\ \hline
Christie's Strategy (+C=3.75$\%$, +[D,C]=0.87$\%$) & 10.9 & 33.4 & 23.8 & \textbf{31.9} & - \\ \hline
\end{tabular}
\caption{Optimal strategies from the poll data + selfish modifications, under imperfect information}
\label{tab:strategy_modifications}
\end{table}

\section{Theoretical Analysis of Strategies: Details} \label{app: sec3details}
\ranklosingc*
\begin{proof}
First, see that under Case-(A), we may encounter a structure where an addition $X$ made to prevent the elimination of a candidate (or a series of such additions) may later lead to that candidate's win. In this case, the total votes at the time of winning may be greater than the Droop quota, allowing transfers of surplus (including the additions) votes to candidates still in the election. Only in this case, do we see optimal addition strategies that include $[\dots, W, \dots ]$, i.e., `altruistic to winner(s)' where other candidates are ranked after the winning candidate. Note that when the candidate needs an addition of votes to win a round, i.e., to reach the Droop quota, there is no surplus to transfer. 

Now, when Case (A) doesn't hold, consider adding $b$ ballots with preferences $[C_1,C_2,.\dots,C_n]$, including at least one winning candidate $C_k$. Let this be part of the optimal addition that achieves the desired set of winning candidates. If $C_k$ is the first winning candidate in this order and is followed by a losing candidate $C_{k+1}$, transferring votes from $C_k$ to $C_{k+1}$ becomes inefficient due to surplus transfer. A more effective strategy would be to use fewer ballots that rank up to $C_k$, securing their win without extra votes. Formally,, divide $b$ votes of type $[C_1,C_2,.\dots,C_n]$ into $b_1$ votes ranking up to $C_k$ and $b_2$ votes ranking from $C_1$ to $C_{k-1}$ and then $C_{k+1}$ onward, with $b_1+b_2$ being less than $b$. This approach allows $C_k$ to win with the exact needed votes, making $b_2$ votes fully effective without reduced value from the surplus transfer. 

If $C_{k+1}$ is also a winner, split the votes into $b_1$ ranking up to $C_k$, $b_2$ up to $C_{k+1}$ (excluding $C_k$), and $b_3$ excluding both $C_k$ and $C_{k+1}$. This might alter the order if $C_{k+1}$ wins before $C_k$, but as all winners in each round eventually make it to the final set (using Lemma \ref{lem: round-election}), the coalition still achieves its goal. This principle is applicable regardless of the sequence of winning candidates. %

Thus, for a coalition aiming for victory, except under Case (A), the most efficient method is to add ballots that primarily rank losing candidates, at least until the second-last choice. 
\end{proof}
\subsection*{Discussion: Optimality of [$\dots, W, \dots$] strategies}
We also discuss the optimality of addition `altruistic to winner(s)' strategies that consider the winning candidate at non-terminal positions. Although Case (A) enables them (which itself is specific), they might still not be practically optimal.

We see that adding $X$ $[\dots, W (C_1), \dots, W (C_2)]$ votes with even two winning candidates (say $C_1$, $C_2$) is optimal only under Case (A) and further, it helps both non-uniformly because of fractional transfers post the first win. Here, $C_1$ gets all $X$ votes, while $C_2$ receives at most $X$ additions: If $C_1$ wins with at least $(Q+X)$ votes in the winning round (should only happen via transfers under Case (A)), then it transfers strategic additions of magnitude $X$. However, $C_2$ receives $X$ only when all $C_1$ votes rank $C_2$ next (due to fractional transfer rules), in which case $C_2$ might want $C_1$ to lose and transfer all such. Thus, in this coalition, $C_1$ `gives away' its additions post-winning without receiving any from $C_2$ (as $C_2$ is still in the contest), while $C_2$ will most benefit if $C_1$ loses, getting all its votes. Thus, if $C_2$ wants to get $X$ votes, it will better benefit from adding $X$ [$L,\dots, L, W$] votes, and letting $C_1$ lose.

Overall, the strategy `altruistic to winner(s)' will be optimal when both $C_1$ and $C_2$ are altruistic to each other and form a coalition to add votes together, but $C_1$ doesn't gain and $C_2$ might strictly prefer bailing out to improve.
\end{document}